\documentclass[11pt,a4paper]{amsart}
\usepackage{amssymb,amsmath,mathtools}
\usepackage{cancel}
\usepackage{palatino}
\usepackage{longtable}
\usepackage{cancel}
\usepackage[all]{xy}
\usepackage{tikz-cd}
\usepackage{color}
\usepackage{accents}
\numberwithin{equation}{section}
\newtheorem{theorem}{Theorem}[section]     
\newtheorem{definition}[theorem]{Definition}

\newtheorem{proposition}[theorem]{Proposition}
\newtheorem{lemma}[theorem]{Lemma}
\newtheorem{rmk}[theorem]{Remark}
\newtheorem{corollary}[theorem]{Corollary}

\newtheorem{remark}[theorem]{Remark}

\def\dna{d_{\nabla}}

\def\na{\nabla}

\def\d{\partial}

\def\f{\frac}
\def\dna{d_{\nabla}}

\def\na{\nabla}

\newcommand{\eqa}{\begin{eqnarray}}
\newcommand{\eeqa}{\end{eqnarray}}
\newcommand{\beq}{\begin{equation}}
\newcommand{\eeq}{\end{equation}}

\newcommand{\derx}{\mathrm{Der}^{\qp}(\hm A)}

\newcommand{\diff}[2]{\frac{\partial #1}{\partial #2}}

\newcommand{\qa}{\alpha}
\newcommand{\qb}{\beta}
\newcommand{\qd}{\delta}
\newcommand{\qg}{\gamma}
\newcommand{\qs}{\sigma}
\newcommand{\qt}{\tau}
\newcommand{\qth}{\theta}
\newcommand{\qe}{\varepsilon}
\newcommand{\qz}{\zeta}
\newcommand{\qp}{\partial}

\newcommand{\qo}{\omega}

\newcommand{\Qg}{\Gamma}
\newcommand{\ql}{\lambda}

\newcommand{\Qd}{\Delta}

\newcommand{\vard}[2]{\frac{\delta #1}{\delta #2}}

\newcommand{\hm}[1]{\hat{\mathcal #1}}

\newcommand{\kk}[1]{\left(#1\right)}

\newcommand{\fk}[2]{\left[#1, #2\right]}
\newcommand{\xb}[1]{\underaccent{\bar}{#1}}

\newcommand{\jja}{\mathcal{\hat A}}
\newcommand{\jjf}{\mathcal{\hat F}}

\numberwithin{equation}{section}

\begin{document}
\title[Generalised Hamiltonian Structures]{Generalised (bi-)Hamiltonian structures of hydrodynamic type and (bi-)flat F-manifolds}

\author{Paolo Lorenzoni}
\address{P.~Lorenzoni:\newline Dipartimento di Matematica e Applicazioni, Universit\`a degli Studi di Milano-Bicocca, 
Via Roberto Cozzi 53, I-20125 Milano, Italy and INFN sezione di Milano-Bicocca}
\email{paolo.lorenzoni@unimib.it}

\author{Zhe Wang}
\address{Z.~Wang:\newline RIKEN Center for Interdisciplinary Theoretical and Mathematical Sciences (iTHEMS), RIKEN, Wako 351-0198, Japan}
\email{zhe.wang.aa@riken.jp}

\date{}

\maketitle

\begin{abstract}
 We introduce the notions of generalised (bi-)Hamiltonian structures which generalise naturally the (bi-)Hamiltonian structures of evolutionary partial differential equations. In the hydrodynamic case, these structures are characterised in terms of geometric data. 
 Furthermore, we show that a generalised (bi)-Hamiltonian structure of hydrodynamic type can be associated with any (bi-)flat F-manifold, and it is compatible with the corresponding principal hierarchy.
\end{abstract}
\tableofcontents

\section{Introduction}
The notion of (local) Hamiltonian operator
\begin{equation}\label{PBHT}
P^{\qa\qb}=g^{\qa\qb}(u(x))\partial_x+\Gamma^{\qa\qb}_{\qg}(u(x))u^\qg_x
\end{equation}
of hydrodynamic type has been introduced by Dubrovin and Novikov in \cite{DN84} (see also \cite{DN89}), here the indices run over $1,\dots, n$ for some fixed $n$ and throughout the whole paper we assume the summation convention over repeated Greek indices. The variable $u(x)$ here is an abbreviation of field variables $u^1(x),\dots,u^n(x)$.

They first observed that the form of such operators is invariant with respect to  change of the field variables. It turns out that the functions $g^{\qa\qb}(u)$ transform as the components of a tensor field of type $(2,0)$ and, assuming $g^{\qa\qb}$  non-degenerate (i.e. assuming $\text{det}(g)\ne 0$) and denoting by $g_{\qa\qb}$ the entries of the inverse, the quantities $\Gamma^{\qg}_{\qa\qb}=-g_{\qa\qs}\Gamma^{\qs\qg}_{\qb}$ transform as the Christoffel symbols of an affine connection. Therefore, properties of such an operator can be studied through differential geometry of these quantities. Imposing the condition
 that the bracket between functionals $F=\int f(u,u_x,...)\,dx$ and $G=\int g(u,u_x,...)\,dx$ defined by
 \begin{equation}
 \{F,G\}=\int\frac{\delta F}{\delta u^\qa}P^{\qa\qb}\frac{\delta G}{\delta u^\qb}\,dx
 \end{equation}
 is a Poisson bracket, they proved that in the non-degenerate case  the tensor field with components $g^{\qa\qb}$ is symmetric, and thus its inverse defines a metric $g$. Furthermore, the affine connection with Christoffel symbols  $\Gamma^{\qg}_{\qa\qb}$ is flat and coincides with the Levi-Civita connection of $g$.

 Bi-Hamiltonian structures of hydrodynamic type have been introduced by Dubrovin in \cite{du93} (see also \cite{du97}). They are defined by pair $(P_0,P_1)$ of compatible Hamiltonian operators of  hydrodynamic type where compatibility is defined in the usual Magri's sense  meaning that the pencil $P_{z}=P_1-z P_0$ is required to be Poisson for any value of the
 parameter $z$. This leads immediately to the geometric counterpart of these structures: flat pencil of metrics. Indeed, due to Dubrovin-Novikov conditions, the pencil $g_{z}=g_1-z g_0$ of contravariant metrics defining $P_0$ and $P_1$ should be flat for any value of the parameter and the pencil of ``contravariant'' Christoffel symbols $\Gamma^{\qa\qb}_{(1)\qg}-z\Gamma^{\qa\qb}_{(0)\qg}$ should coincide with the contravariant Christoffel symbols of the pencil $g_{z}$.

The main examples of bi-Hamiltonian structures come from the theory of Dubrovin-Frobenius manifolds where one of the contravariant metrics is the inverse of the invariant covariant metric of the manifold and the other metric is the so-called intersection form. This includes as special cases Dubrovin-Saito construction on the orbit space of a Coxeter group \cite{saito,SYS,du93,du96} and Dubrovin-Zhang construction \cite{DZ98} on the orbit space of extended affine Weyl groups.
 
Geometry of Hamiltonian and bi-Hamiltonian structures has played a fundamental role in the study of integrable hierarchies, especially in (1) investigating the dispersive deformation of hydrodynamic systems, and (2) understanding their deep relationships to 2d topological field theories. These two aspects are closely related, and it were Dubrovin and Zhang who pioneered the study of (2) in terms of (1) in an axiomatic way \cite{DZ01}. They introduced the notions of integrable hierarchies of topological type to be dispersive evolutionary PDEs that satisfy the following four axioms:
\begin{itemize}
  \item Quasi-triviality: there exists a quasi-Miura type transformation that reduces the whole hierarchy to its dispersionless leading terms.
  \item Tau-symmetric: there exists a formal function of time variables (called the tau-function or the partition function) such that the second order logarithm derivatives solve the PDEs. 
  \item Bi-Hamiltonian: the hierarchies are bi-Hamiltonian whose dispersionless limit is the bi-Hamiltonian structure corresponding to a Dubrovin-Frobenius manifold.
  \item Virasoro symmetries: there exist additional  symmetries of the hierarchy, and these symmetries obey the Virasoro commutation relations.
\end{itemize}
They conjectured that associated to any semisimple 2d topological field theory, there canonically exists a unique integrable hierarchy of topological type such that the partition function of the underlying 2d topological field theory serves as a tau-function of the hierarchy. The main goal of Dubrovin and Zhang's framework is to construct an integrable hierarchy of topological type starting from a (semisimple) Dubrovin-Frobenius manifold. In \cite{DZ01}, they proved the uniqueness of such an integrable hierarchy, and the existence remained open at that time.

The existence of integrable hierarchies of topological type is proved after a deep development of the geometry of bi-Hamiltonian structure over last two decades. The fundamental idea is as follows. The axioms of integrable hierarchies of topological type imply that such a hierarchy must be a deformation of the principal hierarchy of a given Dubrovin-Frobenius manifold. Therefore, one should study integrable bi-Hamiltonian deformations of the principal hierarchy (such that the bi-Hamiltonian axiom is satisfied),  and try to find a particular deformation satisfying all other three axioms. Along this line, it is proved in \cite{DLZ06} that any bi-Hamiltonian deformations satisfies the quasi-triviality axiom by investigating the bi-Hamiltonian cohomology groups $BH^2_2(P_0,P_1)$ and $BH^2_3(P_0,P_1)$. These spaces are modulis of infinitesimal deformations of semisimple bi-Hamiltonian structures, and the computation of these groups relies heavily on the differential geometry of bi-Hamiltonian structures.

A great advancement was made in a series of papers \cite{DLZ13,CPS,CKS}, where almost all cohomology groups $BH^p_d(P_0,P_1)$ are computed for any semisimple bi-Hamiltonian structure $(P_0,P_1)$. It is proved that all infinitesimal deformations are unobstructed and deformations are parametrized by $n$ (= dimension of the underlying manifold) functions of single variable, and these functions are called central invariants.
Using these results, it is proved in \cite{FL12,DLZ18} that the tau-symmetric axiom follows from the condition that all the central invariants being constant. Finally, in \cite{LWZ1,LWZ2,LWZ3,LWZ5}, it is proved that the Virasoro symmetry axiom follows from the condition that all the central invariants are constant and equal, by developing a theory called the variational bi-Hamiltonian cohomology. 

The algebro-geometric counterpart
  of the construction of dispersive deformations of integrable hierarchies of topological type is Givental-Teleman's reconstruction of a cohomological field theory (CohFT) from the associated Dubrovin-Frobenius manifold (see \cite{Gi01,Te12}). In this approach the starting point of constructing topological integrable hierarchies is again a Dubrovin-Frobenius manifold coming from the genus zero part of a CohFT, and the higher genus deformation is reconstructed by the Givental's group action (see \cite{BPS}). As an alternative approach, starting from a CohFT, Buryak introduced in \cite{B15} the double ramification hierarchies, which are Hamiltonian integrable hierarchies whose Hamiltonian densities are defined directly in terms of certain intersection numbers over the double ramification cycle of the moduli space of stable curves. It has been proved in \cite{BLS} that the double ramification hierarchy is Miura-equivalent to the Dubrovin-Zhang topological integrable hierarchy if they are constructed from a same semisimple CohFT. In particular, double ramification hierarchies are bi-Hamiltonian, and an explicit formula for their bi-Hamiltonian structures in terms of the CohFT data is presented in \cite{BRS,BR}.

The main motivation of the present paper is to investigate generalisations of all the constructions above in a more general setting. Indeed, it is observed that many nice properties of Frobenius manifolds (or CohFTs) remain to be valid after weakening some of the requirements of Frobenius manifolds (or CohFTs). This leads to the definition of F-manifold \cite{HM} and F-cohomological field theory (F-CohFT) \cite{ABLR1,ABLR2}. A generalisation of Givental-Teleman's reconstruction in the setting of flat F-manifolds is presented in \cite{ABLR1}, and a construction of double ramification hierarchies for F-CohFT is given in \cite{ABLR2}. However, there is no Dubrovin-Zhang theory for F-manifold and for a good reason: the theory of Dubrovin-Zhang is essentially a theory of geometry of bi-Hamiltonian structure, while there is no natural construction of Hamiltonian structure for flat F-manifolds due to the lack of metric. 

This paper serves as the first step towards establishing a Dubrovin-Zhang axiomatic framework for (bi-)flat F-manifolds. In this paper we introduce a generalisation of (bi-)Hamiltonian structures and study in detail its differential geometry in the hydrodynamic case. Similar to the Hamiltonian operator \eqref{PBHT} of hydrodynamic type, a generalised Hamiltonian structure is defined by an affine flat torsionless connection and by an invertible tensor field $g^{\qa\qb}$ of type $(2,0)$, but the key differences are that $g$ is not required to be symmetric, and the connection is not required to be related to $g$. Moreover, the choice of $g$ is inessential in the sense that different choices of $g$ lead to equivalent generalised Hamiltonian structures. The standard choice of $g$ (i.e. choosing $g$ compatible  with $\nabla$) is related to Gel'fand-Dorfman, Magri-Morosi Hamiltonian formalism on 1-forms \cite{AL12}. This allows defining a generalised Hamiltonian structure for any  system of hydrodynamic type that  can be written  as a system of  conservation laws. According to the results of \cite{Sevennec}, these systems, in the diagonalizable case,  coincide with Tsarev's semi-Hamiltonian integrable systems of hydrodynamic type (\cite{ts91}). 

In this general setting, Magri's compatibility between two generalised Hamiltonian structures turns out to be equivalent to the flatness of a Gauss-Manin type connection associated with the geometric data defining such structures. Using the freedom in the definition of generalised Hamiltonian structures (which is crucial here), it turns out that these data reduce to a pair of flat torsionless connections $\nabla$ and $\nabla^*$ and to a tensor field of type $(1,1)$ with vanishing Nijenhuis torsion satisfying suitable compatibility conditions which are equivalent to the flatness of the associated Gauss-Manin connections. Remarkably, these compatibility conditions are automatically satisfied in the case of bi-flat F-manifolds (see \cite{AL_GM}). Bi-flat F-manifolds are a generalisation of Dubrovin-Frobenius manifolds (any Dubrovin-Frobenius manifold cn be thought as a bi-flat F-manifold).

Many constructions in the theory of Dubrovin-Frobenius manifolds
 and related integrable hierarchies have been extended to bi-flat F-manifolds. For instance,
 \begin{itemize}
 \item The relation with Painlevè transcendents has been investigated in \cite{ALjgp,Limrn,KMS,ALmulti,KM},
 \item The relation with reflection groups has been investigated in \cite{KMS15,ALcomplex,KMS}.
 \item The relation with intesection theory on Deligne-Mumford moduli space of stable curves and related topological hierarchies of double ramification type have been investigated in \cite{ABLR1,ABLR2}. 
 \end{itemize}
However, so far, all the constructions related to the existence of a bi-Hamiltonian structure of hydrodynamic type were out of reach. We think that the results of this paper will open new perspectives in this direction. 

\section{Flat and bi-flat F manifolds}
\subsection{F-manifolds}
F-manifolds have been introduced by Hertling and Manin in \cite{HM}.
\begin{definition}\label{defFmani}
An \emph{F-manifold} is a manifold $M$ equipped with
\begin{itemize}
\item[(i)] a commutative associative bilinear product  $\circ$  on the tangent bundle, satisfying the following identity:
\begin{align}
&[X\circ Y,W\circ Z]-[X\circ Y, Z]\circ W-[X\circ Y, W]\circ Z\label{HMeq1free}\\
&-X\circ [Y, Z \circ W]+X\circ [Y, Z]\circ W +X\circ [Y, W]\circ Z\notag\\
&-Y\circ [X,Z\circ W]+Y\circ [X,Z]\circ W+Y\circ [X, W]\circ Z=0,\notag
\end{align}
for all local vector fields $X,Y,W, Z$, where $[X,Y]$ is the Lie bracket,
\newline
\item[(ii)] a distinguished vector field $e$ on $M$ (called the unit vector field)  such that 
\[e\circ X=X\] 
for all local vector fields $X$. 
\end{itemize}
\end{definition}
Condition \eqref{HMeq1free} is known as the \emph{Hertling-Manin condition}.

\subsection{Flat F-manifolds}
In this paper we will consider F-manifolds equipped with some additional structures.
\begin{definition}\cite{manin}
A flat F-manifold $(M, \circ, \nabla, e)$ is an F-manifold $M$ equipped with a connection
 $\nabla$ satisfying the following conditions:
\begin{itemize}
\item The one parameter family of connections
$$\nabla-\lambda\circ$$
is flat and torsionless for any $\lambda$.
\item
The vector field $e$ is covariantly constant: $\nabla e=0$.
\end{itemize}
\end{definition}

\begin{remark}
It was observed in \cite{ALimrn} that given a flat F-manifold $(M, \circ, \nabla, e)$ (the  existence of the unit is not crucial) the exterior covariant derivatives $d_{\nabla}$ and $d_{\nabla^*}$ associated  with $\nabla$ and $\nabla^*=\nabla+z\circ$
 (where $z$ is a non-vanishing constant) define a differential bicomplex on the space of vector valued differential forms.
 \end{remark}

\subsection{The principal hierarchy}
Given an $n$-dimensional flat F-manifold one can define an integrable hierarchy starting from solutions of the equation  
\begin{equation}\label{symmetries}
d_{\nabla}(X\circ)=0,
\end{equation}
where $d_{\nabla}$ is the exterior covariant derivative defined as
\begin{align}
  \label{AU}
  (\dna &\omega)(X_0, \dots, X_k)=\sum_{i=0}^k (-1)^i \na_{X_i}(\omega(X_0, \dots, \hat{X}_i, \dots, X_k))+\\
  \notag
  & +\sum_{0\leq i<j\leq k}(-1)^{i+j}\omega([X_i, X_j], X_0, \dots, \hat{X}_i, \dots, \hat{X}_j, \dots X_k).
\end{align}
A family of solutions $X_{(\qa,p)}$ $(\qa=1,\dots,n; p=-1,0,1,\dots)$ of the equation \eqref{symmetries} can be constructed as follows: the vector fields $X_{(\qa,-1)}$  are chosen to be covariantly constant with respect to $\nabla$, while other vector fields $X_{(\qa,p)}$ are obtained through the recurrence relation:
\beq\label{recurrenceeq1}
\nabla X_{(\qa,p+1)}=X_{(\qa,p)}\circ. 
\eeq 
If we choose local coordinates $(v^1,\dots,v^n)$ of $M$ and denote by $c^\qa_{\qb\qg}$ the structure constants of the product $\circ$ in these coordinates, then the principal hierarchy associated to $M$ is defined to be the following system of evolutionary PDEs (see \cite{LPR} for details):
\begin{equation}
  \label{AX}
\diff{v^\qa}{t^{\qb,p}}=c^\qa_{\ql\qg}X^\ql_{(\qb,p-1)}v^\qg_x,\quad \qa=1,...,n,\quad p\geq 0.
\end{equation}
These flows commute mutually with each other and define an integrable hierarchy.

\subsection{Bi-flat F-manifolds and Dubrovin-Frobenius manifolds}
Bi-flat F-manifolds are F-manifolds equipped with two ``compatible'' flat structures. They are defined as follows. 
\begin{definition}[\cite{ALjgp,ALmulti}]
A \emph{bi-flat}  F-manifold is a flat $F$-manifold $(M, \circ, \nabla, e)$ equipped with a 
 second distinguished  vector field $E$ (called the \emph{Euler vector field}) satisfying $\mathcal{L}_E \circ=\circ$ and with an additional flat torsionless connection $\nabla^*$ satisfying $\nabla^*E=0$ and the condition
\begin{equation}\label{comp_conn}
(d_{\nabla}-d_{\nabla^{*}})(X\,\circ)=0,
\end{equation}
for any local vector field $X$. 
\end{definition}

\subsection{Dual structure}
Any bi-flat F-manifold $M$ is equipped with a second flat F-manifold structure (usually included in the definition) defined by the data $(M, *, \nabla^*, E)$  where $*$ is the commutative associative product on the tangent bundle defined as 
\begin{equation}\label{dual}
X*Y := (E\circ)^{-1}X\circ Y,
\end{equation}
where $X$ and $Y$ are arbitrary local vector fields and at a generic point the operator $E\circ$ is assumed to be invertible. 

\subsection{A family of Gauss-Manin connections}
 The compatibility of the connections $\nabla$ and $\nabla^*$ defining a bi-flat F-manifold can be formulated in terms of a family of Gauss-Manin connections associated to any bi-flat F-manifold: 
\begin{definition}
Let $(M,L,\nabla,\nabla^*)$ be a manifold equipped with a pair of flat connections
 $\nabla$ and $\nabla^*$ and with a $(1,1)$-tensor field $L$ with vanishing Nijenhuis torsion.  We call Gauss-Manin connections the one-parameter family of connections defined by
\begin{equation}\label{GMnoindex}
\nabla^{GM}_XY=\nabla^*_XY+z(\nabla^*_{L_{z}^{-1}X}Y-\nabla_{L_{z}^{-1}X}Y),
\end{equation}
where $L_z=L-z I$. 
\end{definition}
Denoting by $A^{\qg}_{\qa\qb}$ and $B^{\qg}_{\qa\qb}$ the Christoffel symbols of the connections $\nabla$ and $\nabla^*$ respectively in some local coordinates, the Christoffel symbols of the connection $\nabla^{GM}$ can be written as
\begin{equation}\label{GM}
\Gamma^{\qg}_{\qa\qb}:=B^{\qg}_{\qa\qb}+z\left(L_{z}^{-1}\right)^\rho_\qa\left(B^\qg_{\rho\qb}-A^\qg_{\rho\qb}\right),
\end{equation}
In the framework of integrable systems of hydrodynamic type (this includes semi-Hamiltonian systems) it is natural to assume that
\begin{equation}\label{AHE}
d_{\nabla}(L)=d_{\nabla^*}(L).
\end{equation} 
Indeed, this follows from \emph{almost hydrodynamic equivalence} of the connections
 $\nabla$ and $\nabla^*$ (see \cite{ALimrn}). We have the following theorem.
\begin{theorem}
  \label{BZ}
 Assuming \eqref{AHE}, the vanishing of the curvature $R$ of the Gauss-Manin connection is equivalent to 
\begin{equation}\label{flatnessGM}
\nabla_\qa\left(L^\ql_\qb\Delta^\qg_{\ql\mu}\right)=\nabla_\qb\left(L^\ql_\qa\Delta^\qg_{\ql\mu}\right)
\end{equation}
where $\Delta^\qg_{\qa\qb}=B^\qg_{\qa\qb}-A^\qg_{\qa\qb}$.
\end{theorem}
\begin{proof}. Using \eqref{AHE} one obtains the following identity (see Remark 3.3 in \cite{AL_GM}):
\begin{align*}
R^\qg_{\ql\qa\mu}(L_z)^\ql_\rho(L_z)^\qa_\qb=&\,(R_{\nabla^*})^\qg_{\qa\ql\mu}(L_z)^\ql_\rho(L_z)^\qa_\qb+z^2\left((R_{\nabla^*})^\qg_{\rho\qb\mu}-(R_{\nabla})^\qg_{\rho\qb\mu}\right)\\
&+zL^\ql_\rho\left((R_{\nabla^*})^\qg_{\qb\ql\mu}-(R_{\nabla})^\qg_{\qb\ql\mu}\right)\\
&+zL^\ql_\qb\left((R_{\nabla^*})^\qg_{\rho\ql\mu}-(R_{\nabla})^\qg_{\rho\ql\mu}\right)\\
&+z\left(\nabla_\qb\left(L^\ql_\rho\Delta^\qg_{\ql\mu}\right)-\nabla_\rho\left(L^\ql_\qb\Delta^\qg_{\ql\mu}\right)\right)\\
&-z\left(L_z^{-1}\right)^\qa_\nu(N_L)^\nu_{\qb \rho}\Delta^\qg_{\qa\mu},
\end{align*}
here $N_L$ is the Nijenhuis torsion of $L$ which vanishes by the assumption, and $R_\nabla = R_{\nabla^*} = 0$ by flatness of $\nabla$ and $\nabla^*$. Therefore, the theorem follows.
\end{proof}

Using the relation between generalised bi-Hamiltonian  structures of hydrodynamic type and Gauss-Manin connections associated with the geometric data $(L,\nabla,\nabla^*)$ we  will see in Section 4 that vanishing of Nijenhuis torsion of $L$ and condition \eqref{AHE} are not additional assumptions.

Particularly important is the case  where  the geometric data $(L,\nabla,\nabla^*)$ comes from a bi-flat F-manifold.
\begin{theorem}\label{GMth}\cite{AL_GM}
Let $(M,\circ,\nabla,e,*,\nabla^*,E)$ be a bi-flat F-manifold. 
The  family of Gauss-Manin connections \eqref{GMnoindex} defined by $\nabla$, $\nabla^*$ and  $L=E\circ $ is flat and torsionless for any fixed $z$ on the open set where $L_z$ is invertible. The family above can be further extended by introducing another parameter $\nu$ and
replacing the dual connection with $\nabla^* + \nu *$. The family obtained in this way is still  torsionless and flat.
\end{theorem}
\subsection{A differential bicomplex associated with $(L,\nabla,\nabla^*)$} 
Given the geometric data $(L,\nabla,\nabla^*)$ where $\nabla$ and $\nabla^*$ are flat connections and $L$ is a $(1,1)$-tensor field with vanishing Nijenhuis torsion,  we consider the exterior covariant derivative $d_{\nabla}$ associated with $\nabla$ and the {\em $L$-exterior covariant derivative} $d_{L\nabla^*}$ associated with $\nabla^*$, which is defined as
\[(d_{L\nabla^*}\omega)(X_0, \dots, X_k)=\sum_{i=0}^k (-1)^i \nabla^*_{LX_i}(\omega(X_0, \dots, \hat{X}_i, \dots, X_k))+$$
$$+\sum_{0\leq i<j\leq k}(-1)^{i+j}\omega([X_i, X_j]_L, X_0, \dots, \hat{X}_i, \dots, \hat{X}_j, \dots X_k),\]
here 
\[
[X_i, X_j]_L = [LX_i,X_j]+[X_i,LX_j]-L[X_i,X_j].
\]
From the results of \cite{AL_GM} (see proof of Theorem 4.2) it follows that the pair of differentials $(d_{\nabla},d_{L\nabla^*})$ determines a differential bicomplex structure on the space of vector-valued differential forms if and only if the curvature of the Gauss-Manin connection associated with $(\nabla,\nabla^*,L)$ vanishes. As  a consequence one has the following  theorem.
\begin{theorem}\label{BidifferentialTh}\cite{AL_GM}
On any bi-flat $F$-manifold $(M, \circ, \nabla, e, *, \nabla^*, E)$, the pair of differentials $(d_{\nabla},d_{L\nabla^*})$ determines a differential bicomplex structure on the space of vector-valued differential forms.
\end{theorem}
In other words we have
\[d_{\nabla}^2=0,\qquad  d_{L\nabla^*}^2=0,\qquad d_{\nabla}\cdot d_{L\nabla^*}+d_{L\nabla^*}\cdot d_{\nabla}=0.\]
The first condition follows from the vanishing of the curvature of $\nabla$ and the second condition  follows from the vanishing of  curvature of $\nabla^*$ and  from the vanishing  of the torsion of $L$ (see \cite{ALimrn}). The last condition is equivalent to the vanishing of the curvature of the Gauss-Manin condition.
\begin{remark}
The choice $L=I$ is also interesting. An example is the differential bicomplex associated with a flat F-manifold. 
\end{remark}

\subsection{Dubrovin-Frobenius manifolds, almost dual structure and flat pencil of metrics}
Dubrovin-Frobenius manifolds have been introduced (before bi-flat F-manifolds) by  Dubrovin as a coordinate free formulation of WDVV equations in 2d topological field theories (see \cite{du93}).
Taking into account the definition of bi-flat F-manifolds, Dubrovin-Frobenius manifolds can be defined in the following way.

\begin{definition}
A Dubrovin-Frobenius manifold is a bi-flat F-manifold equipped with 
 a metric $\eta$ compatible with the product $\circ$ and the connection $\nabla$:
\[\langle X\circ Y,Z\rangle=\langle X,Y\circ Z\rangle,\quad \nabla\eta=0,\]
where $\langle\cdot,\cdot\rangle$ is the bilinear form associated with $\eta$ and $X,Y,Z$ are arbitrary local vector fields.
\end{definition}
\noindent
At the points where the operator $E\circ$ is invertible, any Dubrovin-Frobenius manifold is equipped with a second (contravariant) flat metric defined as 
\begin{equation}\label{int_form}
g=E\circ\eta^{-1}.
\end{equation}
It turns out that the data $(g,*,E)$ (where $*$ is defined by \eqref{dual}) satisfy all the
 conditions fulfilled by $(\eta,\circ,e)$  except the fact that the  Euler vector field, which is
  the unit of the dual product, in general is not flat anymore. However, it is possible to deform the Levi-Civita connection of $g$, denoted by $\nabla^{(g)}$, with the dual product
 \[\nabla^*=\nabla^{(g)}+\nu*,\]
 in such a way that $\nabla^*E=0$ for a suitable choice of the constant $\nu$. The connection $\nabla^*$ remains flat and define the dual connection of the associated flat dual structure.

\section{Hamiltonian structures on infinite jet spaces}
Throughout this section, we fix an $n$-dimensional smooth manifold $M$.
\subsection{Infinite jet space of super manifold}
We denote by $\hat M$ the super manifold of dimension $(n|n)$ whose underlying space is $T^*M$ with reversed fibre parity. In local coordinates $(v^1,\dots, v^n)$ of some open set $U\subset M$, we can choose odd coordinates $(\qth_1,\dots,\qth_n)$ for the fibre that are canonically dual to $(v^\qa)$. Note that odd coordinates anti-commute with each other, and for this specific choice they transform covariantly with respect to a coordinate transform on $U$. We denote by $J^\infty(\hat M)$ the infinite jet space associated with $\hat M$ and denote by $\mathcal{\hat A}$ the space of differential polynomials on $J^\infty(\hat M)$. $\mathcal{\hat A}$ is a unital differential algebra with the derivation denoted by $\qp_x\in\Qg(TJ^\infty(\hat M))$. We denote by $\jjf$ the quotient space $\jja/\qp_x\jja$ called the space of local functionals, and by 
\[
\int\colon\jja\to\jjf
\]
the natural projection map. In local coordinates, we have 
\begin{align*}
&\jja = C^\infty(U)\left[v^{\qa,s},\,\qth^s_{\qa}\colon \qa = 1,...,n;\, s\geq 1\right],\\ 
&\qp_x = \sum_{s\geq 0}v^{\qa,s+1}\diff{}{v^{\qa,s}}+\qth_\qa^{s+1}\diff{}{\qth_\qa^s},\quad v^{\qa,0}:=v^\qa,\quad \qth_\qa^0:=\qth_\qa.
\end{align*}
This implies that we can view $v^\qa$ as  a field $v^\qa(x)$ and variables $v^{\qa,s}$ as $\qp_x^s v^\qa(x)$. Therefore, we will also use the notations $v^\qa_x,v^\qa_{xx},\dots$ for $v^{\qa,1},v^{\qa,2},\dots$.
Note that $\jja$ comes with two natural gradations called the super gradation given by 
\[
\deg_\qth v^{\qa,s} = 0,\quad \deg_\qth \qth_\qa^s = 1
\]
and the differential gradation by 
\[
\deg_x v^{\qa,s} = s,\quad \deg_x \qth_\qa^s = s.
\]
The derivation $\qp_x$ is homogeneous with respect to both gradations, and hence $\jjf$ admits induced gradations. We will use $\jja^p$ to denote the subset of homogeneous elements of super degree $p$, and $\jja_d$ the subset of homogeneous elements of differential degree $d$. Denote by $\jja^p_d = \jja^p\cap\jja_d$. Similarly, we define the subspaces $\jjf^p$, $\jjf_d$ and  $\jjf^p_d$.

\subsection{Schouten-Nijenhuis bracket}
The Schouten-Nijenhuis bracket is the bilinear map 
\[
[-,-]\colon \jjf\times\jjf\to\jjf
\]
defined by 
\[
[F,G] = \int \vard{F}{\qth_\qa}\vard{G}{v^\qa}+(-1)^p\vard{F}{v^\qa}\vard{G}{\qth_\qa},\quad F\in\jjf^p,\quad G\in\jjf,
\]
here the variational derivative is defined as usual by 
\[
\vard{F}{v^\qa} = \sum_{s\geq 0} (-\qp_x)^s\diff{f}{v^{\qa,s}},\quad \vard{F}{\qth_\qa} = \sum_{s\geq 0} (-\qp_x)^s\diff{f}{\qth^{s}_\qa},\quad F = \int f.
\]
This bracket defines a graded Lie algebra structure on $\jjf$ (whose gradation is different from the standard notion, see, e.g., \cite{LZJac}) which makes it convenient to study variational multi-vectors.  It induces an important map $D\colon \jjf^p\to \derx^{p-1}$  given by 
\begin{align}
  \label{AB}
D(F) = \sum_{s\geq 0}\qp_x^s\kk{\vard{F}{\qth_\qa}}\diff{}{v^{\qa,s}}+(-1)^p\qp_x^s\kk{\vard{F}{v^\qa}}\diff{}{\qth_\qa^s},\quad F\in\hm F^p,
\end{align}
where we denote by $\derx^{p}$ the space of graded derivations on $\hm A$  of super degree $p$ that commute with $\qp_x$, namely it consists of linear maps $\qd$ satisfying $\qd\circ\qp_x = \qp_x\circ\qd$ and 
\[
\qd(\hm A^q)\subset\hm A^{q+p},\quad \qd(fg) = \qd(f)g+(-1)^{pq}f\qd(g),\quad \forall\,f\in\hm A^q,\quad g\in\hm A.
\]  
The space $\derx^{p}$ also admits the differential gradation induced from that of $\hm A$, and we denote by $\derx^p_d$ the subset of homogeneous elements of differential degree $d$. Denote by $\derx = \bigoplus_p \derx^p$, and hence we extend the definition of $D$ by linearity to obtain $D\colon \hm F\to \derx$.   
This is a graded Lie algebra homomorphism in the sense that 
\begin{equation}
  \label{AC}
  D\kk{[F,G]} = (-1)^{p-1}[D(F),D(G)],\quad F\in\jjf^p,\quad G\in\jjf,
\end{equation}
here the bracket on the right-hand side is the usual graded commutator of derivations. The map $D$ is injective, and the elements in the subspace $D(\hm F)\subset\derx$ are called derivations of $D$-type in \cite{LWZ4}.

The morphism $D$ provides a comparison between the space of local functionals and the space of derivations, and is fundamental in the construction of variational (bi)-Hamiltonian cohomologies, leading to the final proof of the existence of Dubrovin-Zhang's integrable hierarchies of topological type. In particular, the study of derivations which are not of $D$-type becomes essential, and a generalisation of techniques that previously apply to $\hm F$ to the larger space $\derx$ would enable one to investigate wilder range of integrable hierarchies. The results presented in this paper arise from this idea.

\subsection{$\qth$-formalism of Hamiltonian structure}
A Hamiltonian structure is an element $P\in\jjf^2$ such that $[P,P] = 0$. A bi-Hamiltonian structure is a pair $(P_0,P_1)$ such that $P_z = P_1-z P_0$ is a Hamiltonian structure for any $z$.

This notion of Hamiltonian structure is equivalent to the usual formulation using Hamiltonian operators. Indeed, for any local functional $P\in\jjf^2$, we set 
\[
\vard{P}{\qth_\qa} = \sum_{s\geq 0} P^{\qa\qb}_s\qth_\qb^s,\quad P^{\qa\qb}_s\in \jja^0,
\]
then for functionals 
\[
F = \int f\left(v^\qa,v^{\qa}_x,v^\qa_{xx},\dots\right),\quad G =  \int g\left(v^\qa,v^{\qa}_x,v^\qa_{xx},\dots\right),
\]
a straightforward computation shows that 
\[
\{F,G\}_P := \int\sum_{s\geq 0}\vard{F}{v^\qa}P^{\qa\qb}_s\qp_x^s\vard{G}{v^\qb} = [[F,P],G].
\]
Hence, it follows that the bracket $\{-,-\}_P$ is anti-commutative, and it satisfies the Jacobi identity if and only if $[P,P] = 0$. Conversely, given a Poisson bracket $\{-,-\}_P$ of the above form, we can form a Hamiltonian structure by 
\begin{equation}
  \label{BK}
P = \frac 12\int  \sum_{s\geq 0} P^{\qa\qb}_s\qth_\qa\qth_\qb^s.
\end{equation}

Using the notions above, we can define (bi-)Hamiltonian systems as follows.
\begin{definition}
  \label{BJ}
  Let $X\in \hm F^1$. $X$ is called a Hamiltonian system if there exists a Hamiltonian structure $P$ and a local functional $H\in\hm F^0$ such that $X = -[P,H]$. $X$ is called a bi-Hamiltonian system if $X = -[P_0,H_0] = -[P_1,H_1]$ for some bi-Hamiltonian structure $(P_0,P_1)$ and local functionals $X_0,X_1\in\hm F^0$.
\end{definition}
This definition coincides with the usual notion of Hamiltonian systems. A Hamiltonian system is usually defined to be evolutionary PDEs of the form 
\[
\diff{v^\qa}{t} = \sum_{s\geq 0}P^{\qa\qb}_s\qp_x^s\kk{\vard{H}{v^\qb}}
\]
for some Hamiltonian operator $\sum_{s\geq 0}P^{\qa\qb}_s\qp_x^s$ and some local Hamiltonian functional $H\in\hm F^0$, then it is easy to see that 
\[
D_{-[P,H]}(v^\qa) = \diff{v^\qa}{t},
\]
where $P\in\hm F^2$ is the Hamiltonian structure given by \eqref{BK}. Therefore, the Hamiltonian system above is equivalent to the local functional $X = -[P,H]$.

\subsection{Hamiltonian structure of hydrodynamic type}
Let $P\in\jjf^2_1$ be a local functional with
\[
\vard{P}{\qth_\qa} = g^{\qa\qb}\qth_\qb^1+\Qg^{\qa\qb}_\qg v^\qg_x\qth_\qb,\quad g^{\qa\qb},\,\Qg^{\qa\qb}_\qg \in\jja^0_0.
\]
Such a local functional is called a (non-degenerate) Hamiltonian structure of hydrodynamic type if $P$ is a Hamiltonian structure and $\det(g^{\qa\qb})\neq 0$. We denote by $g_{\qa\qb}$ the inverse of $g^{\qa\qb}$.
\begin{theorem}\cite{DN84}
  \label{AR}
  $P$ is a Hamiltonian structure of hydrodynamic type if and only if $g = (g_{\qa\qb})$ defines a flat metric on $M$ and the quantities $\Gamma^{\qg}_{\qa\qb}=-g_{\qa\qs}\Gamma^{\qs\qg}_{\qb}$ give the Christoffel symbols of the Levi-Civita connection of $g$.
\end{theorem}
By choosing local flat coordinate systems of $g$, a Hamiltonian structure of hydrodynamic type can always be written as the form 
\[
P = \frac12\int\eta^{\qa\qb}\qth_\qa\qth_\qb^1,
\]
for some symmetric non-degenerate constant matrix $\eta^{\qa\qb}$.

One of the most important examples of bi-Hamiltonian structures of hydrodynamic type are those associated with Dubrovin-Frobenius manifolds. Indeed, the invariant metric $\eta$ and the intersection form $g$ defines a flat pencil, which corresponds to a bi-Hamiltonian structure of hydrodynamic type.

\section{Main results}
Throughout this section, we fix an $n$-dimensional smooth manifold $M$ and local coordinates $(v^1,\dots,v^n;\qth_1,\dots,\qth_n)$ of $\hat M$.
\subsection{Generalised Hamiltonian structure}
\begin{definition}
  A generalised Hamiltonian structure is an odd derivation 
  \[
  \diff{}{\qt}\in\derx^1
  \]
  such that
  \[
  \fk{\diff{}{\qt}}{\diff{}{\qt}} = 0.
  \]
  A generalised bi-Hamiltonian structure is a pair of generalised Hamiltonian structures 
  \[
  \diff{}{\qt_0},\,\diff{}{\qt_1}\in\derx^1
  \]
  such that
  \[
  \diff{}{\qt_1}-z \diff{}{\qt_0}
  \]
  remains to be a generalised Hamiltonian structure for any $z$.
\end{definition}

A Hamiltonian structure $P$ gives rise to a generalised Hamiltonian structure $D(P)$ by \eqref{AB} due to the fact \eqref{AC}, and similarly a bi-Hamiltonian structure gives a generalised bi-Hamiltonian structure. However, not all generalised Hamiltonian structures arise in this way. For example, we fix a constant matrix $(\eta^{\qa\qb})$, then it is easy to check that 
\[
\diff{v^\qa}{\qt} = \eta^{\qa\qb}\qth_\qb^1,\quad \diff{\qth_\qa}{\qt} = 0
\] 
is a generalised Hamiltonian structure. It corresponds to a genuine Hamiltonian structure if and only if $\eta^{\qa\qb} = \eta^{\qb\qa}$.

\subsection{Poisson bracket on variational 1-forms}
Let us consider the space of variational differential 1-forms $\mathcal E$ (we use the notations in \cite{LWZ2}). Recall that, in local coordinates $(v^1,\dots,v^n)$, $\mathcal E$ consists of elements of the form
\begin{equation}
    \label{BG}
\qo = \int f_\qa\qd v^\qa,\quad f_\qa\in\hm A^0.
\end{equation}
Note that we have an isomorphism of vector spaces
\[
\mathcal E\cong\derx^{-1},
\]
where an element of the form \eqref{BG} gives rise to a derivation of super degree $-1$ by 
\[
\diff{v^\qa}{t_\qo} = 0,\quad \diff{\qth_\qa}{t_\qo} = -f_\qa.
\]
Therefore, we also call an odd derivation in $\derx^{-1}$ a variational 1-form.

Assume that we are given a generalised Hamiltonian structure $\diff{}{\qt}\in\derx^1$. Let us define a bilinear map 
\begin{equation}
    \label{BH}
    \{-,-\}_\qt\colon \derx^{-1}\times \derx^{-1}\to \derx^{-1},\quad \{X,Y\} = [X_\qt,Y],
\end{equation}
here and henceforth in this subsection, we denote by 
\[
X_\qt = \fk{X}{\diff{}{\qt}} =\fk{\diff{}{\qt}}{X},\quad X\in\derx^{-1}. 
\]

\begin{theorem}
The bracket \eqref{BH} endows the space $\derx^{-1}$ a Lie algebra structure for any generalised Hamiltonian structure $\diff{}{\qt}$.
\end{theorem}
\begin{proof}
    By the graded Jacobi identity satisfied by the graded commutators defined on $\derx$, we have 
    \[
    \fk{\fk{\diff{}{\qt}}{X}}{Y}+\fk{[X,Y]}{\diff{}{\qt}}+\fk{\fk{Y}{\diff{}{\qt}}}{X} = 0,\quad X,Y\in\derx^{-1}.
    \] 
    Note that $[X,Y]\in\derx^{-2}$, but such a derivation must be zero (see Lemma 8 of \cite{LWZ3}), hence we arrive at 
    \[
    \fk{X_\qt}{Y}+\fk{Y_\qt}{X} = 0,
    \]
    in other words the bracket \eqref{BH} is antisymmetric. We also note that the above identity is equivalent to 
    \begin{equation}
        \label{BI}
        \fk{X_\qt}{Y} = \fk{X}{Y_\qt}.
    \end{equation}

    It only remains to verify the Jacobi identity for the bracket \eqref{BH}. Let us fix arbitrary $X,Y,Z\in\derx^{-1}$. It follows that 
    \begin{align*}
        \{X,\{Y,Z\}\} &=\fk{X_\qt}{\{Y,Z\}}\\
        &=-\fk{\fk{X}{\{Y,Z\}}}{\diff{}{\qt}}-\fk{\fk{\{Y,Z\}}{\diff{}{\qt}}}{X}\\
        &=-\fk{\fk{\diff{}{\qt}}{\{Y,Z\}}}{X},
    \end{align*}
    here we use again the graded Jacobi identity on $\derx$ and the fact that $\fk{X}{\{Y,Z\}} = 0$ due to  the degree reason. We also use a slightly simplified notation $\{-,-\}$ instead of $\{-,-\}_\qt$.
    Furthermore, we have
    \begin{align*}
        \fk{\diff{}{\qt}}{\{Y,Z\}} &= \fk{\diff{}{\qt}}{[Y_\qt,Z]}\\
        &=\fk{Y_\qt}{\fk{Z}{\diff{}{\qt}}}+\fk{Z}{\fk{\diff{}{\qt}}{Y_\qt}}\\
        &=\fk{Y_\qt}{Z_\qt},
    \end{align*} 
    here we use the condition that $\diff{}{\qt}$ is a generalised Hamiltonian structure.
    We combine the results above, and arrive at
    \begin{align*}
        &\{X,\{Y,Z\}\}+\{Y,\{Z,X\}\}+\{Z,\{X,Y\}\}\\
        =&\,-\fk{\fk{Y_\qt}{Z_\qt}}{X}-\fk{\fk{Z_\qt}{X_\qt}}{Y}-\fk{\fk{X_\qt}{Y_\qt}}{Z}\\
        =&\,\fk{\fk{Z_\qt}{X}}{Y_\qt}+\fk{\fk{X}{Y_\qt}}{Z_\qt}+\fk{\fk{X_\qt}{Y}}{Z_\qt}+\fk{\fk{Y}{Z_\qt}}{X_\qt}\\
        &+\fk{\fk{Y_\qt}{Z}}{X_\qt}+\fk{\fk{Z}{X_\qt}}{Y_\qt}\\
        =&\,-2\fk{Y_\qt}{\fk{Z_\qt}{X}}-2\fk{Z_\qt}{\fk{X_\qt}{Y}}-2\fk{X_\qt}{\fk{Y_\qt}{Z}}\\
        =&\,-2\{Y,\{Z,X\}\}-2\{Z,\{X,Y\}\}-2\{X,\{Y,Z\}\},
    \end{align*}
    which implies 
    \[
    \{X,\{Y,Z\}\}+\{Y,\{Z,X\}\}+\{Z,\{X,Y\}\} = 0,
    \]
    note that we use no more than the graded Jacobi identity on $\derx$ and the identity \eqref{BI}. The theorem is proved.

\end{proof}

The Lie bracket \eqref{BH} on $\derx^{-1}$ induces a Lie algebra structure on $\mathcal E$. If the generalised Hamiltonian structure is given by a Hamiltonian structure of hydrodynamic type, this bracket coincides with the one defined in \cite{AL12}. Indeed, let us fix a constant (symmetric) matrix $\eta^{\qa\qb}$ and introduce a Hamiltonian structure 
\[
P = \frac12\int\eta^{\qa\qb}\qth_\qa\qth_\qb^1,
\] 
which corresponds to the
generalised Hamiltonian structure
\[
\diff{v^\qa}{\qt} = \eta^{\qa\qb}\qth_\qb^1,\quad \diff{\qth_\qa}{\qt} = 0.
\]
Pick $X,Y\in\derx^{-1}$ given by 
\[
X(v^\qa) = 0,\quad X(\qth_\qa) = f_\qa;\quad Y(v^\qa) = 0, \quad Y(\qth_\qa) = g_\qa,
\]
note that they correspond to variational 1-forms 
\[
\qo_X = -\int f_\qa\qd v^\qa,\quad \qo_Y = -\int g_\qa\qd v^\qa. 
\]
Let us compute the bracket $\{X,Y\}$. Firstly, we have 
\begin{align*}
    X_\qt(v^\qa) = X\kk{\eta^{\qa\qb}\qth_\qb^1} = \eta^{\qa\qb}\qp_x f_\qb,\quad X_\qt(\qth_\qa) = \diff{f_\qa}{\qt} = \sum_{k\geq 0}\diff{f_\qa}{v^{\qb,k}}\eta^{\qb\mu}\qth_\mu^{k+1}.
\end{align*}
Therefore, it follows that 
\begin{align*}
    \{X,Y\}(\qth_\qa) &= X_\qt\kk{Y(\qth_\qa)}-Y\kk{X_\qt\kk{\qth_\qa}}\\
    &=X_\qt\kk{g_\qa}-\sum_{k\geq 0}\diff{f_\qa}{v^{\qb,k}}\eta^{\qb\mu}\qp_x^{k+1}g_\mu\\
    &=\sum_{k\geq 0}\diff{g_\qa}{v^{\qb,k}}\eta^{\qb\mu}\qp_x^{k+1} f_\mu-\sum_{k\geq 0}\diff{f_\qa}{v^{\qb,k}}\eta^{\qb\mu}\qp_x^{k+1}g_\mu,
\end{align*}
and hence the variational 1-from corresponding to $\{X,Y\}$ is given by  
\[
\int \kk{\eta^{\qb\mu}\sum_{k\geq 0}\kk{\qp_x^{k+1}g_\mu}\diff{f_\qa}{v^{\qb,k}}-\kk{\qp_x^{k+1}f_\mu}\diff{g_\qa}{v^{\qb,k}}}\qd v^\qa,
\]
which coincides with the expression presented in Proposition 3.4 of \cite{AL12}.

The original motivation of \cite{AL12} to introduce the notion of Poisson bracket on variational 1-forms is to generalise the usual definitions \eqref{BJ} of Hamiltonian systems and study the systems of the form 
\begin{equation}
  \label{BW}
\diff{v^\qa}{t} = \sum_{s\geq 0}P^{\qa\qb}_s\qp_x^s\qo^\qb,
\end{equation}
where $\sum_{s\geq 0}P^{\qa\qb}_s\qp_x^s$ is still a Hamiltonian operator, but $\qo^\qb$ may not necessarily be given by 
\[
\qo^\qb = \vard{H}{v^\qb}
\]
for some Hamiltonian functional $H$. A further generalisation leads to the following definitions. 
\begin{definition}
  \label{BL}
  A derivation $X\in\derx^0$ is called a generalised Hamiltonian system if there exists a generalised Hamiltonian structure $\diff{}{\qt}\in\derx^{1}$ and a variational 1-form $\qo\in\derx^{-1}$ such that 
  \begin{equation}
    \label{BX}
  X = \fk{\diff{}{\qt}}{\qo}.
  \end{equation}
  It is called a generalised bi-Hamiltonian system if 
  \[
  X = \fk{\diff{}{\qt_0}}{\qo}=\fk{\diff{}{\qt_1}}{\qz}
  \]
  for some generalised bi-Hamiltonian structure $(\diff{}{\qt_0},\diff{}{\qt_1})$ and variational 1-forms $\qo$ and $\qz$.
\end{definition} 

\begin{proposition}
  Fix a generalised Hamiltonian structure $\diff{}{\qt}$ and two variational 1-forms $\qo, \qz$. Denote by 
  \[
  X = \fk{\diff{}{\qt}}{\qo},\quad Y = \fk{\diff{}{\qt}}{\qz}.
  \]
  If $\{\qo,\qz\}_\qt = 0$, then $[X,Y] = 0$.
\end{proposition}
\begin{proof}
  This follows directly from the graded Jacobi identity on $\derx$.
\end{proof}

This is a generalisation of the well-known fact that commuting Hamiltonian functionals lead to commuting Hamiltonian systems.

\subsection{Differential geometry of generalised Hamiltonian structure of hydrodynamic type}
Let us consider an odd derivation $\diff{}{\qt}\in\derx^1_1$ given by 
\begin{equation}
  \label{AD} 
  \diff{v^\qa}{\qt} = g^{\qa\qb}\qth_\qb^1+\Qg^{\qa\qb}_\qg v^\qg_x\qth_\qb,\quad \diff{\qth_\qa}{\qt} = V_\qa^{\mu\ql}\qth_\ql\qth_\mu^1+Q^{\mu\ql}_{\qa\qg}v^\qg_x\qth_\ql\qth_\mu,
\end{equation}
where the quantities $g^{\qa \qb}$, $\Qg^{\qa\qb}_\qg$, $V^{\qa\qb}_\qg$ and $Q^{\ql\mu}_{\qa\qb}$ are some functions on $M$, and we assume that 
\[
Q^{\ql\mu}_{\qa\qb}+Q^{\mu\ql}_{\qa\qb} = 0.
\]

\begin{definition}
  An odd derivation of hydrodynamic type is that of the form \eqref{AD} such that 
  $\det(g^{\qa\qb})\neq 0$. If it defines a generalised Hamiltonian structure, we call it a (non-degenerate) generalised Hamiltonian structure of hydrodynamic type.
\end{definition}

\begin{lemma}
  An odd derivation of hydrodynamic type of the form \eqref{AD} defines a bundle isomorphism
  \[
  g^\sharp\colon T^*M\to TM,\quad g^\sharp(dv^\qa) = g^{\qb\qa}\qp_\qb
  \]
  and a connection $\prescript{\vee}{}\nabla$ on the cotangent bundle given by 
  \[
  \prescript{\vee}{}\nabla_{\qp_\ql}(dv^\qa) = \prescript{\vee}{}A^\qa_{\ql\mu} dv^\mu,\quad \prescript{\vee}{}A^\qa_{\ql\mu}g^{\qb\mu} = \Qg^{\qb\qa}_\ql.
  \]
\end{lemma}
\begin{proof}
  We need to show that under a coordinate transformation on $M$, the quantities $g^{\qa\qb}$ transform as a $(2,0)$-tensor and $\prescript{\vee}{}A^\qa_{\ql\mu}$ transform as the Christoffel symbols of a connection on the cotangent bundle.

  Let us choose two local coordinates $(v^\qa)$ and $(\bar v^\qa)$ on $M$ whose Jacobian matrices are denoted by 
  \[
\bar J^\qa_\qb =\diff{\bar v^\qa}{v^\qb},\quad \xb J^\qa_\qb =\diff{v^\qa}{\bar v^\qb}.
\]
Then we see that the corresponding (canonical) odd variables $\qth_\qa$ and $\bar\qth_\qa$ are related by 
\[
\qth_\qa = \bar J^\qb_\qa\bar\qth_\qb.
\]
Denote by 
\[
\diff{\bar v^\qa}{\qt} = \bar g^{\qa\qb}\bar \qth_\qb^1+\bar \Qg^{\qa\qb}_\qg \bar v^\qg_x\bar \qth_\qb,
\]
then it follows that
\begin{align*}
    \diff{\bar v^\qa}{\qt} =&\,\bar J^\qa_\qb\kk{g^{\qb\ql}\qth_\ql^1+\Qg^{\qb\ql}_\qg v^\qg_x\qth_\ql} \\
    =&\,\bar J^\qa_\qb g^{\qb\ql}\qp_x\kk{\bar J^\mu_\ql\bar\qth_\mu}+ \bar J^\qa_\qb\Qg^{\qb\ql}_\qg \xb J^\qg_\qd \bar v^\qd_x \bar J^\mu_\ql\bar\qth_\mu\\
    =&\, \bar J^\qa_\qb g^{\qb\ql}\bar J^\mu_\ql\bar\qth_\mu^1+\bar J^\qa_\qb g^{\qb\ql}\qp_\rho \bar J^\mu_\ql v^\rho_x\bar\qth_\mu+ \bar J^\qa_\qb\Qg^{\qb\ql}_\qg \xb J^\qg_\qd \bar v^\qd_x \bar J^\mu_\ql\bar\qth_\mu\\
    =&\, \bar J^\qa_\qb g^{\qb\ql}\bar J^\mu_\ql\bar\qth_\mu^1+\bar J^\qa_\qb g^{\qb\ql}\qp_\rho \bar J^\mu_\ql \xb J^\rho_\qd \bar v^\qd_x\bar\qth_\mu+ \bar J^\qa_\qb\Qg^{\qb\ql}_\qg \xb J^\qg_\qd \bar v^\qd_x \bar J^\mu_\ql\bar\qth_\mu,
\end{align*}
from which we conclude that 
\begin{equation}
    \label{AE}
\bar g^{\qa\mu} = \bar J^\qa_\qb g^{\qb\ql}\bar J^\mu_\ql,\quad \bar\Qg^{\qa\mu}_\qd = \bar J^\qa_\qb\bar J^\mu_\ql \xb J^\qg_\qd   \Qg^{\qb\ql}_\qg+\bar J^\qa_\qb \qp_\rho \bar J^\mu_\ql \xb J^\rho_\qd g^{\qb\ql}.
\end{equation}
The first identity of \eqref{AE} implies that $g^{\qa\qb}$ defines a tensor field of type $(2,0)$, and thus can be interpreted as a bundle map $T^*M\to TM$, which is an isomorphism due to non-degeneracy assumption. A straightforward computation using the second identity of \eqref{AE} shows that the quantities $\prescript{\vee}{}A^\qa_{\ql\mu}$ determined by 
\[
\prescript{\vee}{}A^\qa_{\ql\mu}g^{\qb\mu} = \Qg^{\qb\qa}_\ql
\]
transform as Christoffel symbols of a connection on the cotangent bundle. The lemma is proved.
\end{proof}

For an odd derivation of hydrodynamic type of the form \eqref{AD}, we can define the transported affine connection $\nabla$ by 
\[
\nabla_X(Y) = g^\sharp \prescript{\vee}{}\nabla_X(g^\flat (Y)),\quad\forall\, X,Y\in\Qg(TM),
\]
here we denote by $g^\flat\colon TM\to T^*M$ the inverse of $g^\sharp$. If we set 
\[
g^\flat(\qp_\qa) = g_{\qa\qb}dv^\qb,
\]
then it follows that 
\begin{equation}
  \label{AN}
g^{\mu\qa}g_{\mu\qb} = \qd^\qa_\qb,\quad g_{\qa\mu}g^{\qb\mu} = \qd^\qb_\qa,
\end{equation}
and the Christoffel symbols of the transported connection read 
\begin{equation}
  \label{AM}
\nabla_{\qp_\qa}\qp_\qb = A^\ql_{\qa\qb}\qp_\ql,\quad A^\ql_{\qa\qb} = -g_{\qb\mu}\qp_\qa g^{\ql\mu}+g_{\qb\qg}\Qg^{\ql\qg}_\qa,
\end{equation}
which is equivalent to 
\begin{equation}
  \label{BC}
  \Qg^{\qa\qb}_\qg = \qp_\qg g^{\qa\qb}+g^{\mu\qb}A^\qa_{\qg\mu}.
\end{equation}

\begin{definition}
  We call $(g^\sharp,\nabla)$ the geometric data associated to an odd derivation of hydrodynamic type of the form \eqref{AD}.
\end{definition}

\begin{lemma}
  \label{AF}
  A generalised Hamiltonian structure of hydrodynamic type is uniquely determined by its associated geometric data.
\end{lemma}
\begin{proof}
  We fix a generalised Hamiltonian structure of hydrodynamic type of the form \eqref{AD}. We only need to show that the quantities $V^{\qa\qb}_\qg$ and $Q^{\mu\ql}_{\qa\qb}$ are determined uniquely from the associated geometric data.

Let us denote by 
\[
W^\qa:=\diff{}{\qt}\diff{v^\qa}{\qt},\quad Z_\qa:=\diff{}{\qt}\diff{\qth_\qa}{\qt}.
\]
Then it follows from a straightforward computation that 
  \begin{align*}
    W^\qa=&\,\diff{}{\qt}\kk{g^{\qa\qb}\qth_\qb^1+\Qg^{\qa\qb}_\qg v^\qg_x\qth_\qb}\\
    =&\,\qp_\qe g^{\qa\qb}\kk{g^{\qe\rho}\qth_\rho^1+\Qg^{\qe\rho}_\qd v^\qd_x\qth_\rho}\qth_\qb^1+g^{\qa\qb}\qp_x\kk{ V_\qb^{\mu\ql}\qth_\ql\qth_\mu^1+Q^{\ql\mu}_{\qb\qg}v^\qg_x\qth_\ql\qth_\mu}\\
    &+\qp_\qe\Qg^{\qa\qb}_\qg\kk{g^{\qe\rho}\qth_\rho^1+\Qg^{\qe\rho}_\qd v^\qd_x\qth_\rho}v^\qg_x\qth_\qb+\Qg^{\qa\qb}_\qg\qp_x\kk{g^{\qg\rho}\qth_\rho^1+\Qg^{\qg\rho}_\qd v^\qd_x\qth_\rho}\qth_\qb\\
    &+\Qg^{\qa\qb}_\qg v^\qg_x\kk{V_\qb^{\mu\ql}\qth_\ql\qth_\mu^1+Q^{\ql\mu}_{\qb\qg}v^\qg_x\qth_\ql\qth_\mu},\\
    Z_\qa =&\,\diff{}{\qt}\kk{V_\qa^{\mu\ql}\qth_\ql\qth_\mu^1+Q^{\ql\mu}_{\qa\qg}v^\qg_x\qth_\ql\qth_\mu}\\
    =&\,\qp_\qe V^{\mu\ql}_\qa\kk{g^{\qe\rho}\qth_\rho^1+\Qg^{\qe\rho}_\qd v^\qd_x\qth_\rho}\qth_\ql\qth_\mu^1+V^{\mu\ql}_\qa\kk{V_\ql^{\qb\qg}\qth_\qg\qth_\qb^1+Q_{\ql\rho}^{\qb\qg}v^\rho_x\qth_\qb\qth_\qg}\qth_\mu^1\\
    &-V^{\mu\ql}_\qa\qth_\ql\qp_x\kk{V_\mu^{\qb\qg}\qth_\qg\qth_\qb^1+Q_{\mu\rho}^{\qb\qg}v^\rho_x\qth_\qb\qth_\qg}+\diff{}{\qt}\kk{Q^{\ql\mu}_{\qa\qg}v^\qg_x\qth_\ql\qth_\mu}.
\end{align*}
We have $W^\qa = 0$ and $Z_\qa = 0$ by assumption, and the identity 
\[
\frac{\qp^2W^\qa}{\qp\qth_\mu^2\qp\qth_\ql} = 0
\]
implies that 
\begin{equation}
  \label{AL}
g^{\qa\qb}V_\qb^{\mu\ql} = \Qg^{\qa\ql}_\qg g^{\qg\mu},
\end{equation}
while the identity
\[
\frac{\qp^3 W^\qa}{\qp\qth_\ql\qp\qth_\mu\qp v^\qd_{xx}} = 0
\]
implies that 
\begin{equation}
  \label{AZ}
g^{\qa\qb}Q^{\ql\mu}_{\qb\qd} = \frac12\kk{\Qg^{\qa\ql}_\qg\Qg^{\qg\mu}_\qd-\Qg^{\qa\mu}_\qg\Qg^{\qg\ql}_\qd}.
\end{equation}
Indeed, the quantities $V^{\qa\qb}_\qg$ and $Q^{\mu\ql}_{\qa\qb}$ are determined uniquely from the geometric data since we assume that $g^{\qa\qb}$ is invertible. The lemma is proved.
\end{proof}

\begin{theorem}
  \label{AQ}
  Fix an odd derivation of hydrodynamic type with the associated geometric data $(g^\sharp,\nabla)$ of the form \eqref{AD} satisfying the identities \eqref{AL} and \eqref{AZ}. Then it defines a generalised Hamiltonian structure if and only if $\nabla$ is torsionless and flat. 
\end{theorem}
\begin{proof}
  For the "only if" part, we continue to use the same notations as those in the proof of Lemma \ref{AF}. 
  The identity
\[
\frac{\qp^2W^\qa}{\qp\qth_\mu^1\qp\qth_\ql^1} = 0
\]
implies that 
\begin{equation}
    \label{AK}
\qp_\qe g^{\qa\mu}g^{\qe\ql}-\qp_\qe g^{\qa\ql}g^{\qe\mu}+g^{\qa\qb}V^{\mu\ql}_\qb -g^{\qa\qb}V^{\ql\mu}_\qb = 0. 
\end{equation}
Using the identity \eqref{AL} and the definition \eqref{AM} of the transported connection, it is easy to see that \eqref{AK} is equivalent to 
\[
A^\ql_{\qa\qb} = A^\ql_{\qb\qa},
\]
hence the connection $\nabla$ is torsionless. Next we consider the identity
\[
\frac{\qp^3Z_\qa}{\qp\qth_\ql\qp\qth_\mu^1\qp\qth_\qg^1} = 0,
\]
from which it follows that 
\begin{equation}
\label{AO}
-\qp_\qe V^{\mu\ql}_\qa g^{\qe\qg}+\qp_\qe V^{\qg\ql}_\qa g^{\qe\mu}+V_\qa^{\mu\rho}V_\rho^{\qg\ql}-V_\qa^{\qg\rho}V_\rho^{\mu\ql}-V_\qa^{\rho\ql}V_\rho^{\mu\qg}+V_\qa^{\rho\ql}V_\rho^{\qg\mu} = 0.
\end{equation}
Note that the identity \eqref{AL} is equivalent to 
\begin{equation}
  \label{AP}
V^{\qa\ql}_\qb = g^{\qg\qa}\prescript{\vee}{}A^\ql_{\qg\qb},
\end{equation}
hence we conclude from the identity \eqref{AO} that 
\begin{align*}
    &-\qp_\qe V^{\mu\ql}_\qa g^{\qe\qg} +\qp_\qe V^{\qg\ql}_\qa g^{\qe\mu}\\
    =&\,  -\qp_\qe\kk{g^{\rho\mu}\prescript{\vee}{}A^\ql_{\rho\qa}}g^{\qe\qg} +\qp_\qe \kk{g^{\rho\qg}\prescript{\vee}{}A^\ql_{\rho\qa}} g^{\qe\mu}\\
=&\, -\qp_\qe\kk{g^{\rho\mu}}\prescript{\vee}{}A^\ql_{\rho\qa}g^{\qe\qg} +\qp_\qe \kk{g^{\rho\qg}}\prescript{\vee}{}A^\ql_{\rho\qa} g^{\qe\mu}\\
&-g^{\rho\mu}\qp_\qe \prescript{\vee}{}A^\ql_{\rho\qa}g^{\qe\qg} + g^{\rho\qg}\qp_\qe \prescript{\vee}{}A^\ql_{\rho\qa} g^{\qe\mu}\\
=&\,\prescript{\vee}{}A^\ql_{\rho\qa}\kk{g^{\rho\qb}V^{\mu\qg}_\qb -g^{\rho\qb}V^{\qg\mu}_\qb}-g^{\rho\mu}\qp_\qe \prescript{\vee}{}A^\ql_{\rho\qa}g^{\qe\qg} + g^{\rho\qg}\qp_\qe \prescript{\vee}{}A^\ql_{\rho\qa} g^{\qe\mu}\\
=&\,V^{\rho\ql}_\qa V^{\mu\qg}_\rho-V^{\rho\ql}_\qa V^{\qg\mu}_\rho-g^{\rho\mu}\qp_\qe \prescript{\vee}{}A^\ql_{\rho\qa}g^{\qe\qg} + g^{\rho\qg}\qp_\qe \prescript{\vee}{}A^\ql_{\rho\qa} g^{\qe\mu},
\end{align*}
where we have used \eqref{AK} for the third equality. Using \eqref{AP} again, it is easy to see that the connection $\prescript{\vee}{}\nabla$ is flat, namely we have 
\[
\prescript{\vee}{}\nabla_{\qp_\qa}\prescript{\vee}{}\nabla_{\qp_\qb} dv^\ql = \prescript{\vee}{}\nabla_{\qp_\qb}\prescript{\vee}{}\nabla_{\qp_\qa} dv^\ql.
\] 
This implies that $\nabla$ is also flat, indeed, the Riemann curvature tensor of $\prescript{\vee}{}\nabla$ is related to that of $\nabla$ by
\[
R_{\nabla} = g^\sharp\circ R_{\prescript{\vee}{}\nabla}\circ g^\flat.
\]
This finishes the proof of the ``only if'' part.

For the ``if'' part, we can define a flat connection $\prescript{\vee}{}\nabla$ on $T^*M$ by
\[
\prescript{\vee}{}\nabla_X(\qo) = g^\flat\,\nabla_X\kk{g^\sharp(\qo)},\quad \forall\,X\in\Qg(TM),\quad \qo\in\Qg(T^*M).
\]
We choose locally a flat coordinate system $(v^\qa)$ of $\prescript{\vee}{}\nabla$ on $M$ such that 
    \[
    \prescript{\vee}{}\nabla dv^\qa = 0,\quad g^\sharp(dv^\qa) = g^{\qb\qa}\qp_\qb.
    \]
    In terms of this coordinate system, we see that 
    \[
\nabla_{\qp_\qa}\qp_\qb = -g_{\qb\mu}\qp_\qa g^{\ql\mu}\qp_\ql.
    \]
    By using the identities \eqref{AL} and \eqref{AZ}, we arrive at an odd derivation of hydrodynamic type given by 
    \[
    \diff{v^\qa}{\qt} = g^{\qa\qb}\qth_\qb^1,\quad \diff{\qth_\qa}{\qt} = 0,
    \]
    whose associated geometric data is exactly $(g^\sharp,\nabla)$. Then the condition of it being a generalised Hamiltonian structure reads 
    \[
    \qp_\qg g^{\qa\qb}g^{\qg\mu} = \qp_\qg g^{\qa\mu}g^{\qg\qb},
    \]
    which is equivalent to the torsionless condition 
    \[
   \nabla_{\qp_\qa}\qp_\qb = \nabla_{\qp_\qb}\qp_\qa.
    \]
    The theorem is proved.
\end{proof}

By comparing Theorem \ref{AR} and Theorem \ref{AQ}, for a generalised Hamiltonian structure of hydrodynamic type whose associated geometric data is $(g^\sharp,\nabla)$, we conclude that it arises from a Hamiltonian structure if and only if $g^\sharp$ is symmetric and $\nabla g^\sharp = 0$. However, the data $g^\sharp$ and $\nabla$ are generally not related, and it seems unnatural that there are no constraints on the choice of $g^\sharp$. This issue can be explained by considering different choices of frames of fibers of $\hat M$. Indeed, for local coordinates $(v^\qa)$ of $M$, we can freely choose any (local) frames for the fibers of $\hat M$, and different choices are related by (local) isomorphisms of $T^*M$. 

Let us fix an arbitrary choice $(v^\qa;\qth_\qa)$ of local coordinates of $\hat M$ and fix any odd derivation of hydrodynamic type of the form \eqref{AD} with the associated geometric data $(g^\sharp,\nabla)$. For an arbitrary isomorphism $T$ of $T^*M$ given by 
\[
T(dv^\qa) = T^\qa_\qb dv^\qb,
\]
we can define new local coordinates $(v^\qa;\qs_\qa)$ on $\hat M$ by setting 
\begin{equation}
  \label{AS}
\qth_\qa = T_\qa^\qb\qs_\qb.
\end{equation}
Then it follows that 
\begin{align*}
  \diff{v^\qa}{\qt} = g^{\qa\qb}\qth_\qb^1+\Qg^{\qa\qb}_\qg v^\qg_x\qth_\qb=\tilde g^{\qa\qb}\qs_\qb^1+\tilde \Qg^{\qa\qb}_\qg v^\qg_x\qs_\qb,
\end{align*}
with
\[
\tilde g^{\qa\qb} = g^{\qa\rho}T^\qb_\rho,\quad \tilde\Qg^{\qa\qb}_\qg = g^{\qa\rho}\qp_\qg T^\qb_\rho+\Qg^{\qa\rho}_\qg T^\qb_\rho.
\]
A straightforward computation shows that 
\[
-g_{\qb\mu}\qp_\qa g^{\ql\mu}+g_{\qb\qg}\Qg^{\ql\qg}_\qa = -\tilde g_{\qb\mu}\qp_\qa \tilde g^{\ql\mu}+\tilde g_{\qb\qg}\tilde \Qg^{\ql\qg}_\qa,
\]
therefore, after the transformation \eqref{AS}, the odd derivation 
\[
\diff{v^\qa}{\qt}=\tilde g^{\qa\qb}\qs_\qb^1+\tilde \Qg^{\qa\qb}_\qg v^\qg_x\qs_\qb
\]
has the associated geometric data $(g^\sharp\circ T,\nabla)$. 

The discussions above immediately imply the following two statements.
\begin{proposition}
  \label{AW}
Assume that  we are given two generalised Hamiltonian structures whose associated geometric data are $(g^\sharp,\nabla)$ and $(\tilde g^\sharp,\nabla)$ respectively. Then these two generalised Hamiltonian structures are equivalent in the sense that they can be related by a certain change of odd variables of the form \eqref{AS}.
\end{proposition}

\subsection{Generalised Hamiltonian hydrodynamic system and flat F-manifolds}
A hydrodynamic evolutionary PDE is a system of the form 
\begin{equation}  \label{AJ}
\diff{v^\qa}{t} = X^\qa_\qb v_x^\qb.
\end{equation}
It is easy to check that the quantities $X^\qa_\qb$ define a tensor field $X$ of type $(1,1)$ and hence we view it as a differential 1-form with vector values. It is called to be compatible with a generalised Hamiltonian structure $\diff{}{\qt}$ of hydrodynamic type if there exists  an extension of $\diff{}{t}$ to an element in $\derx^0_1$ by  
\begin{equation}
  \label{AV}
\diff{v^\qa}{t} = X^\qa_\qb v_x^\qb,\quad \diff{\qth_\qa}{t} = Y^\qb_\qa\qth_\qb^1+M^\qb_{\qa\qg}v^\qg_x\qth_\qb,
\end{equation}
such that 
\[
\fk{\diff{}{\qt}}{\diff{}{t}} = 0.
\]

\begin{theorem}
  \label{AY}
  Fix a generalised Hamiltonian structure $\diff{}{\qt}$ with the associated geometric data $(g^\sharp,\nabla)$. The hydrodynamic system \eqref{AJ} is compatible with $\diff{}{\qt}$ if and only if $d_\nabla(X) = 0$. 
\end{theorem}
\begin{proof}
  The proof is similar to that of Lemma \ref{AF} and Theorem \ref{AQ}, and is given by a straightforward computation. For the ``only if'' part, we set 
  \[
\diff{v^\qa}{\qt} = g^{\qa\qb}\qth_\qb^1+\Qg^{\qa\qb}_\qg v^\qg_x\qth_\qb,\quad \diff{\qth_\qa}{\qt} = V_\qa^{\mu\ql}\qth_\ql\qth_\mu^1+Q^{\mu\ql}_{\qa\qg}v^\qg_x\qth_\ql\qth_\mu,
\]
and we know there exists an extension of the form \eqref{AV}
such that 
\[
\fk{\diff{}{\qt}}{\diff{}{t}} = 0.
\]
Let us denote by 
\[
Z^\qa = \diff{}{\qt}\diff{v^\qa}{t},\quad W^\qa = \diff{}{t}\diff{v^\qa}{\qt},
\]
then they are given by 
\begin{align*}
    Z^\qa =&\,\diff{}{\qt}\kk{X^\qa_\qb v^\qb_x}\\
    =&\,\qp_\mu X^\qa_\qb\kk{g^{\mu\ql}\qth_\ql^1+\Qg^{\mu\ql}_\qd v^\qd_x\qth_\ql}v^\qb_x+X^\qa_\qb\qp_x\kk{g^{\qb\ql}\qth_\ql^1+\Qg^{\qb\ql}_\qd v^\qd_x\qth_\ql},\\
    W^\qa =&\, \diff{}{t}\kk{g^{\qa\qb}\qth_\qb^1+\Qg^{\qa\qb}_\qg v^\qg_x\qth_\qb}\\
    =&\,\qp_\mu g^{\qa\qb}X^\mu_\ql v^\ql_x\qth_\qb^1+g^{\qa\qb}\qp_x\kk{Y^\ql_\qb\qth_\ql^1+M^\ql_{\qb\qg}v^\qg_x\qth_\ql}\\
    &+\qp_\mu\Qg^{\qa\qb}_\qg X^\mu_\ql v^\ql_x v^\qg_x\qth_\qb+\Qg^{\qa\qb}_\qg\qp_x\kk{X^\qg_\mu v^\mu_x}\qth_\qb+\Qg^{\qa\qb}_\qg v^\qg_x\kk{Y^\ql_\qb \qth_\ql^1+M^\ql_{\qb\mu}v^\mu_x\qth_\ql}.
\end{align*}
The identity 
\[
\diff{Z^\qa}{\qth_\ql^2} = \diff{W^\qa}{\qth_\ql^2}
\]
implies that 
\begin{equation}
  \label{AT}
g^{\qa\qb}Y^\ql_\qb = X^\qa_\qb g^{\qb\ql},
\end{equation}
and the identity 
\[
\frac{\qp^2 Z^\qa}{\qp v^\qg_{xx}\qp\qth_\ql} = \frac{\qp^2 W^\qa}{\qp v^\qg_{xx}\qp\qth_\ql}
\]
implies that 
\begin{equation}
  \label{BB}
X^\qa_\qb\Qg^{\qb\ql}_\qg = g^{\qa\qb}M^\ql_{\qb\qg}+\Qg^{\qa\ql}_{\qb}X^\qb_\qg.
\end{equation}
This shows that a hydrodynamic system is uniquely determined from the quantities $X^\qa_\qb$.
We then consider the equation 
\[
\frac{\qp^2 Z^\qa}{\qp v^\qb_x\qp \qth_\ql^1} = \frac{\qp^2 W^\qa}{\qp v^\qb_x\qp \qth_\ql^1}
\]
which, after a straightforward computation using \eqref{AT} and \eqref{AM}, is equivalent to 
\[
\qp_\qd X^\qa_\qb+ A^\qa_{\mu\qd} X^\mu_\qb = \qp_\qb X^\qa_\qd+A^\qa_{\qb\mu}X^\mu_\qd.
\]
This can be written in an invariant way, using the definition \eqref{AU} of exterior covariant differential, as 
\[
d_{\nabla} X = 0.
\]

The ``if'' part is proved also similarly to that of Theorem \ref{AQ}. By choosing  suitable local coordinates, we may assume 
    \[
    \diff{v^\qa}{\qt} = g^{\qa\qb}\qth_\qb^1,\quad \diff{\qth_\qa}{\qt} = 0,\quad X = X^\qa_\qb \qp_\qa\otimes dv^\qb.
    \]
    We can uniquely extend the system \eqref{AJ} by using the identities \eqref{AT} and \eqref{BB}, and arrive at   
    \[
    \diff{v^\qa}{t} = X^\qa_\qb v_x^\qb,\quad \diff{\qth_\qa}{t} = Y^\qb_\qa\qth_\qb^1,\quad g^{\qa\qb}Y^\ql_\qb = X^\qa_\qb g^{\qb\ql},
    \]
    then a straightforward computation shows that the condition 
    \[
\fk{\diff{}{\qt}}{\diff{}{t}} = 0
\]
is equivalent to $d_\nabla(X) = 0$. The theorem is proved. 
\end{proof}

\begin{corollary}
  \label{BV}
Any system of hydrodynamic type that can be written as a system of conservation laws admits a generalised Hamiltonian structure of hydrodynamic type.
\end{corollary}
\begin{proof} 
  Recall that the system \eqref{AJ} is called a system of conservation laws if in some local coordinates $(u^\qa)$, the flows \eqref{AJ} are of the form 
  \[
  \diff{u^\qa}{t} = \qp_x(V^\qa) = \qp_\qb V^\qa v^\qb_x.
  \]
  Then in these local coordinates, the components of the tensor $X$ read 
  \[
    X^\qa_\qb = \qp_\qb V^\qa.
  \]
  We denote by $\nabla$ the affine flat torsionless connection uniquely defined by the condition $X=d_{\nabla}V$. Clearly, the connection $\nabla$ is well-defined and, by definition, coordinates $(u^\alpha)$ serve as 
  flat coordinates of $\nabla$. The result follows from the identity $d^2_{\nabla}=0$.
\end{proof}

\begin{rmk}
  In the case of Corollary \ref{BV}, we can further choose a contravariant flat metric $\eta$ compatible with $\nabla$. With this choice the geometric data $(\eta^\sharp,\nabla)$ defines a Poisson bracket of hydrodynamic type on 1-forms, and the corresponding system is then of the form \eqref{BW} (see \cite{AL12}).
\end{rmk}

Clearly any generalised Hamiltonian system, as defined by \eqref{BX}, is compatible with the corresponding generalised Hamiltonian structure. The converse is true for hydrodynamic case.
\begin{theorem}
    Assume we have a flow $\diff{}{t}\in\derx^0_1$ and a generalised Hamiltonian structure $\diff{}{\qt}\in\derx^1_1$ of hydrodynamic type. Then 
    \begin{equation}
        \label{BM}
    \fk{\diff{}{t}}{\diff{}{\qt}} = 0
    \end{equation}
    if and only if there exists a variational 1-form $H\in\derx^{-1}_0$ such that 
    \begin{equation}
        \label{BN}
    \diff{}{t} = \fk{\diff{}{\qt}}{H}.
    \end{equation}
\end{theorem}
\begin{proof}
The ``if'' part is obvious, and we only need to verify the ``only if'' part.

By choosing suitable local coordinates on $\hat M$, we may assume 
\begin{align*}
 \diff{v^\qa}{t} &= X^\qa_\qb v_x^\qb,\quad \diff{\qth_\qa}{t} = Y^\qb_\qa\qth_\qb^1,\quad g^{\qa\qb}Y^\ql_\qb = X^\qa_\qb g^{\qb\ql},\\
 \diff{v^\qa}{\qt} &=g^{\qa\qb}\qth_\qb^1,\quad \diff{\qth_\qa}{\qt} = 0.
\end{align*}
Then the condition \eqref{BM} reads
\[
\qp_\rho X^\qa_\qb-g_{\mu\ql}\qp_\rho g^{\qa\ql} X^\mu_\qb = \qp_\qb X^\qa_\rho-g_{\mu\ql}\qp_\qb g^{\qa\ql} X^\mu_\rho,
\]
which can be equivalently written as 
\begin{equation}
    \label{BO}
    \qp_\rho\kk{X^\ql_\qa g_{\ql\qb}} = \qp_\qa\kk{X^\ql_\rho g_{\ql\qb}}. 
\end{equation}
We need to find a variational 1-form $H\in\derx^{-1}_0$ satisfying \eqref{BN}. Set
\[
H(v^\qa) = 0,\quad H(\qth_\qa) = H_\qa,
\]
then the equation \eqref{BN} is equivalent to the following system of PDEs:
\[
\qp_\qa H_\qb = X^\ql_\qa g_{\ql\qb}.
\]
The conditions \eqref{BO} guarantee the existence of solutions $H_\qa$ to the above PDE system. The theorem is proved.
\end{proof}

Now we can apply the above results to the principal hierarchy of a flat F-manifold.

\begin{theorem}
  \label{BD}
  Let $(M, \circ, \nabla, e)$ be a flat F-manifold and fix an arbitrary bundle isomorphism $g^\sharp\colon T^*M\to TM$. Then the geometric data $(g^\sharp,\nabla)$ defines a generalised Hamiltonian structure $\diff{}{\qt}$ of hydrodynamic type. Moreover, the principal hierarchy associated to $M$ is a generalised Hamiltonian system with respect to  $\diff{}{\qt}$.
\end{theorem}
\begin{proof}
  The geometric data $(g^\sharp,\nabla)$ defines a generalised Hamiltonian structure $\diff{}{\qt}$ is obvious from the fact that $\nabla$ is torsionless and flat.

  Let us fix any local coordinates $(v^\qa)$ of $M$, then it follows from definition \eqref{AX} that any flow of the principal hierarchy associated to $M$ is of the form 
  \begin{equation}
    \label{BA}
  \diff{v^\qa}{t} = c^\qa_{\qb\qg}X^\qg v^\qb_x,
  \end{equation}
  here $c^\qa_{\qb\qg}$ are the structure constants of $\circ$ given by $\qp_\qb\circ \qp_\qg = c^\qa_{\qb\qg}\qp_\qa$, and $X = X^\qa \qp_\qa$ is a vector field satisfying
  \[
  d_\nabla(X\circ) = 0.
  \]
  Since the $(1,1)$-tensor corresponding to the system \eqref{BA} is exactly $X\circ$, hence it is a generalised Hamiltonian system due to Theorem \ref{AY}. The theorem is proved.
\end{proof}

\subsection{Generalised bi-Hamiltonian structures and bi-flat F-manifolds} Now we focus on the differential geometry of generalised bi-Hamiltonian structure.
Consider a generalised bi-Hamiltonian structure $(\diff{}{\qt_0},\diff{}{\qt_1})$ of hydrodynamic type, the geometric data associated to $\diff{}{\qt_0}$ is denoted by $(\eta^\sharp,\nabla)$ and by $(g^\sharp,\nabla^*)$ for that associated to $\diff{}{\qt_1}$.
\begin{definition}
We call $(L,\nabla,\nabla^*)$ the geometric data associated to the generalised bi-Hamiltonian structure $(\diff{}{\qt_0},\diff{}{\qt_1})$, where $L$ is the tensor field of $(1,1)$ type defined by
\[
L = g^\sharp(\eta^\sharp)^{-1}.
\]
\end{definition}
Due to Proposition \ref{AW}, though the geometric data associated to a hydrodynamic generalised Hamiltonian structure depends on the choice of odd variables, that of a generalised bi-Hamiltonian structure is independent of such choices.

\begin{proposition}
  \label{BY}
  Let $(L,\nabla,\nabla^*)$ be the geometric data associated to a generalised bi-Hamiltonian structure $(\diff{}{\qt_0},\diff{}{\qt_1})$. Then the following properties hold true:
\begin{itemize}
  \item[(1)]$d_\nabla(L) = d_{\nabla^*}(L)$;
  \item[(2)] The Nijenhuis torsion $N$ of $L$ vanishes;
  \item[(3)] $(\nabla_X (L\Qd))(Y,Z) = (\nabla_Y (L\Qd))(X,Z)$ for any (local) vector fields $X$,$Y$,$Z$, here $\Qd$ is the tensor field of type $(1,2)$ given by $\Qd(Y,Z) = \nabla^*_YZ-\nabla_YZ$. In terms of  local coordinates, this is equivalent to \eqref{flatnessGM}.
\end{itemize}
\end{proposition}
\begin{proof}
  All the desired properties are tensorial, hence it suffices to check them in arbitrary local coordinates. Hence, we may choose suitable local coordinates and assume the generalised bi-Hamiltonian structure reads
  \begin{align*}
&\diff{v^\qa}{\qt_0} = \eta^{\qa\qb}\qth_\qb^1,\quad \diff{\qth_\qa}{\qt_0} = 0,\\
&\diff{v^\qa}{\qt_1} = g^{\qa\qb}\qth_\qb^1+\Qg^{\qa\qb}_\qg v^\qg_x\qth_\qb,\quad \diff{\qth_\qa}{\qt_1} = V_\qa^{\mu\ql}\qth_\ql\qth_\mu^1+Q^{\mu\ql}_{\qa\qg}v^\qg_x\qth_\ql\qth_\mu,\quad Q^{\mu\ql}_{\qa\qg}=-Q^{\ql\mu}_{\qa\qg}.
  \end{align*}
  Denote by $\eta_{\qa\qb}$ the inverse of $\eta^{\qa\qb}$ satisfying
  \[
  \eta^{\qa\mu}\eta_{\qb\mu} = \eta^{\mu\qa}\eta_{\mu\qb} = \qd^\qa_\qb,
  \]
  and hence the tensor $L$ is given by 
  \[
  L(\qp_\qa) = L^\qb_\qa\qp_\qb,\quad L^\qb_\qa = g^{\qb\mu}\eta_{\qa\mu}.
  \]
  If we denote by $A^\qg_{\qa\qb}$ and $B^\qg_{\qa\qb}$ the Christoffel symbols of $\nabla$ and $\nabla^*$ respectively, we have 
 \[
    A^\qg_{\qa\qb} = -\eta_{\qa\mu}\qp_\qb\eta^{\qg\mu},\quad \Qg^{\qa\qb}_\qg = \qp_\qg g^{\qa\qb}+g^{\mu\qb}B^\qa_{\mu\qg}.
\]
  By a straightforward computation, it follows that  
  \begin{align*}
    W^\qa =&\, \fk{\diff{}{\qt_0}}{\diff{}{\qt_1}}v^\qa\\
    =&\,\diff{}{\qt_0}\kk{g^{\qa\qb}\qth_\qb^1+\Qg^{\qa\qb}_\qg v^\qg_x\qth_\qb}+\diff{}{\qt_1}\kk{\eta^{\qa\qb}\qth_\qb^1}\\
    =&\,\qp_\mu g^{\qa\qb}\eta^{\mu\ql}\qth_\ql^1\qth_\qb^1+\qp_\mu\Qg^{\qa\qb}_\qg\eta^{\mu\ql}v^\qg_x\qth_\ql^1\qth_\qb+\Qg^{\qa\qb}_\qg\qp_x\kk{\eta^{\qg\ql}\qth_\ql^1}\qth_\qb\\
    &+\qp_\mu\eta^{\qa\qb}\kk{g^{\mu\ql}\qth_\ql^1+\Qg^{\mu\ql}_\qg v^\qg_x\qth_\ql}\qth_\qb^1\\
    &+\eta^{\qa\qb}\qp_x\kk{V_\qb^{\mu\ql}\qth_\ql\qth_\mu^1+Q^{\ql\mu}_{\qb\qg}v^\qg_x\qth_\ql\qth_\mu},
  \end{align*}
  and we know $W^\qa = 0$ by assumption. Then the identity 
  \[
  \diff{W^\qa}{v^\qg_{xx}} = 0
  \]
  implies 
  \[
  Q^{\mu\ql}_{\qa\qg} = 0,
  \]
  and the identity
  \[
  \frac{\qp^2 W^\qa}{\qp \qth_\ql\qp \qth_\mu^2} = 0
  \]
  implies 
  \begin{equation}
    \label{BP}
    \eta^{\qa\qb}V_\qb^{\mu\ql} = \Qg^{\qa\ql}_\qg \eta^{\qg\mu},
  \end{equation}
and the identity 
\[
\frac{\qp^2 W^\qa}{\qp\qth_\ql^1\qp\qth_\mu^1} = 0
\]
implies 
\begin{equation*}
\qp_\rho g^{\qa\mu}\eta^{\rho\ql}+\qp_\rho\eta^{\qa\mu}g^{\rho\ql}+\eta^{\qa\qb}V^{\mu\ql}_\qb = 
\qp_\rho g^{\qa\ql}\eta^{\rho\mu}+\qp_\rho \eta^{\qa\ql}g^{\rho\mu}+\eta^{\qa\qb}V^{\ql\mu}_\qb,
\end{equation*}
which is equivalent, after multiplying the both sides with $\eta_{\qz\ql}\eta_{\qd\mu}$ and using the identity \eqref{BP}, to 
\begin{equation}
    \label{BQ}
\eta_{\qd\mu}\qp_\qz g^{\qa\mu}-L^\rho_\qz A^\qa_{\rho\qd}+\Qg^{\qa\ql}_\qd\eta_{\qz\ql} = \eta_{\qz\mu}\qp_\qd g^{\qa\mu}-L^\rho_\qd A^\qa_{\rho\qz}+\Qg^{\qa\ql}_\qz\eta_{\qd\ql}. 
\end{equation}
In what follows, we will frequently use the following two identities:
\begin{align}
  \label{DA}
\eta_{\qd\mu}\qp_\qz g^{\qa\mu}&=\qp_\qz L^\qa_\qd - L^\qa_\mu A^\mu_{\qz\qd},\\
\label{DB}
\Qg^{\qa\ql}_\qd\eta_{\qz\ql}&=\qp_\qd L^\qa_\qz-L^\qa_\ql A^\ql_{\qd\qz}+L^\rho_\qz B^\qa_{\rho\qd},
\end{align}
both of which can be verified straightforwardly. It follows that \eqref{BQ} is equivalent to 
\begin{equation}
  \label{BR}
  L^\rho_\qz \Qd^\qa_{\rho\qd} = L^\rho_\qd \Qd^\qa_{\rho\qz}, 
\end{equation}
where we denote by $\Qd^\qa_{\rho\qd} = B^\qa_{\rho\qd}-A^\qa_{\rho\qd}$. Therefore, it follows that 
\begin{align*}
  d_{\nabla^*}(L)(\qp_\qa,\qp_\qb)-d_{\nabla}(L)(\qp_\qa,\qp_\qb) = L^\ql_\qb \Qd^\rho_{\qa\ql}-L^\ql_\qa \Qd^\rho_{\qb\ql} = 0,
\end{align*}
and hence the first property holds true.

Next, we verify that $L$ has vanishing Nijenhuis torsion. Recall that the Nijenhuis torsion $N$ of $L$ is a vector valued 2-form whose components read 
\[
N^\qg_{\qa\qb} = L^\mu_\qa\qp_\mu L^\qg_\qb-L^\mu_\qb\qp_\mu L^\qg_\qa-L^\qg_\mu\qp_\qa L^\mu_\qb+L^\qg_\mu\qp_\qb L^\mu_\qa.
\]
By comparing the two identities \eqref{BP} and \eqref{AL}, we arrive at 
\[
\Qg^{\qa\ql}_\qg g^{\qg\mu}g_{\qa\qb} = \Qg^{\qa\ql}_\qg \eta^{\qg\mu}\eta_{\qa\qb},
\]
and after multiplying both sides with $g^{\rho\qb}\eta_{\qz\mu}$ and taking the summation over repeated indices, this is equivalent to 
\[
L^\qg_\qz \Qg^{\rho\ql}_\qg = L^\rho_\qg\Qg^{\qg\ql}_\qz.
\]
Multiply both sides with $\eta_{\qa\ql}$ and use identities \eqref{DA} and \eqref{DB}, we arrive at 
\begin{equation*}
  \label{BS}
L^\qg_\qz\qp_\qg L^\rho_\qa-L^\rho_\qg\qp_\qz L^\qg_\qa = L^\qg_\qz L^\rho_\ql A^\ql_{\qg\qa}-L^\rho_\qg L^\qg_\ql A^\ql_{\qa\qz}+L^\rho_\qg L^\mu_\qa B^\qg_{\mu\qz}-L^\qg_\qz L^\mu_\qa B^\rho_{\mu\qg}.
\end{equation*}
Taking into account the definition of the Nijenhuis tensor, we arrive at 
\[
N^\rho_{\qz\qa} = \kk{L^\qg_\qz\qp_\qg L^\rho_\qa-L^\rho_\qg\qp_\qz L^\qg_\qa}-\kk{L^\qg_\qa\qp_\qg L^\rho_\qz-L^\rho_\qg\qp_\qa L^\qg_\qz}:=J+K,
\]
where we denote 
\begin{align*}
  J &= L^\qg_\qz L^\rho_\ql A^\ql_{\qg\qa}-\cancel{L^\rho_\qg L^\qg_\ql A^\ql_{\qa\qz}}-L^\qg_\qa L^\rho_\ql A^\ql_{\qg\qz}+\cancel{L^\rho_\qg L^\qg_\ql A^\ql_{\qz\qa}}\\
  &=L^\rho_\ql\kk{L^\qg_\qz  A^\ql_{\qg\qa}-L^\qg_\qa  A^\ql_{\qg\qz}},
\end{align*}
and 
\begin{align*}
  K&=L^\rho_\qg L^\mu_\qa B^\qg_{\mu\qz}-\cancel{L^\qg_\qz L^\mu_\qa B^\rho_{\mu\qg}}-L^\rho_\qg L^\mu_\qz B^\qg_{\mu\qa}+\cancel{L^\qg_\qa L^\mu_\qz B^\rho_{\mu\qg}}\\
  &=L^\rho_\ql \kk{L^\qg_\qa B^\ql_{\qg\qz}- L^\qg_\qz B^\ql_{\qg\qa}}.
\end{align*}
By using \eqref{BR}, we have 
\[
N^\rho_{\qz\qa} = J+K = L^\rho_\ql\kk{L^\qg_\qa \Qd^\ql_{\qg\qz}-L^\qg_\qz\Qd^\ql_{\qg\qa}} = 0.
\]
Therefore, we verify that the Nijenhuis torsion $N$ of $L$ vanishes.

Finally, we prove the third identity. We consider the identity 
\[
\frac{\qp^3 W^\qa}{\qp v^\qg_x\qp\qth_\ql\qp\qth_\mu^1} = 0,
\]
and it implies that 
\[
\eta^{\qa\qb}\qp_\qg V^{\mu\ql}_\qb+\qp_\rho\eta^{\qa\mu}\Qg^{\rho\ql}_\qg = \qp_\qg\eta^{\rho\mu}\Qg^{\qa\ql}_\rho+\qp_\rho\Qg^{\qa\ql}_\qg \eta^{\rho\mu}.
\]
Multiplying both sides with $\eta_{\qz\ql}\eta_{\qe\mu}$ and using 
\[
\eta^{\qa\qb}\qp_\qg V^{\mu\ql}_\qb = \qp_\qg\kk{\eta^{\rho\mu}\Qg^{\qa\ql}_\rho}-\qp_\qg\eta^{\qa\qb}\eta_{\rho\qb}\eta^{\qd\mu}\Qg^{\rho\ql}_\qd = \qp_\qg\kk{\eta^{\rho\mu}\Qg^{\qa\ql}_\rho}+A^\qa_{\qg\rho}\eta^{\qd\mu}\Qg^{\rho\ql}_\qd
\]
which follows from \eqref{BP}, we arrive at 
\begin{equation*}
\eta_{\qz\ql}\eta_{\qe\mu}\qp_\qg\kk{\eta^{\rho\mu}\Qg^{\qa\ql}_\rho}+\eta_{\qz\ql}A^\qa_{\qg\rho}\Qg^{\rho\ql}_\qe-\eta_{\qz\ql}A^\qa_{\rho\qe}\Qg^{\rho\ql}_\qg = -\eta_{\qz\ql}A^\rho_{\qe\qg}\Qg^{\qa\ql}_\rho+\eta_{\qz\ql}\qp_\qe\Qg^{\qa\ql}_\qg.
\end{equation*}
Using identities \eqref{DA} and \eqref{DB}, after a lengthy computation, we arrive at 
\begin{align*}
  0=&\, \qp_\qg\kk{L^\rho_\qz B^\qa_{\rho\qe}}-\qp_\qe\kk{L^\rho_\qz B^\qa_{\rho\qg}}+\qp_\qe L^\rho_\qz A^\qa_{\qg\rho}-\qp_\qg L^\rho_\qz A^\qa_{\rho\qe}\\
  &+L^\rho_\ql\Qd^\qa_{\rho\qg}A^\ql_{\qe\qz}-L^\rho_\ql \Qd^\qa_{\rho\qe}A^\ql_{\qg\qz}+L^\rho_\qz\Qd^\ql_{\rho\qe}A^\qa_{\qg\ql}+L^\rho_\qz A^\ql_{\rho\qe}A^\qa_{\qg\ql}\\
  &-L^\rho_\qz\Qd^\ql_{\rho\qg}A^\qa_{\ql\qe}-L^\rho_\qz A^\ql_{\rho\qg}A^\qa_{\ql\qe}.
\end{align*}
Furthermore, it follows that
\begin{align*}
  \qp_\qe L^\rho_\qz A^\qa_{\qg\rho}-\qp_\qg L^\rho_\qz A^\qa_{\rho\qe} &=\qp_\qe\kk{ L^\rho_\qz A^\qa_{\qg\rho}}-\qp_\qg \kk{L^\rho_\qz A^\qa_{\rho\qe}}- L^\rho_\qz\kk{\qp_\qe A^\qa_{\qg\rho}-\qp_\qg A^\qa_{\rho\qe}}\\
  &=\qp_\qe\kk{ L^\rho_\qz A^\qa_{\qg\rho}}-\qp_\qg \kk{L^\rho_\qz A^\qa_{\rho\qe}}- L^\rho_\qz\kk{A^\ql_{\qe\rho}A^\qa_{\qg\ql}-A^\ql_{\rho\qg}A^\qa_{\qe\ql}},
\end{align*}
where for the second equality we use the flatness of $\nabla$. Hence, using \eqref{BR} it follows that
\begin{align*}
  0=&\, \qp_\qg\kk{L^\rho_\qz \Qd^\qa_{\rho\qe}}-\qp_\qe\kk{L^\rho_\qz \Qd^\qa_{\rho\qg}}\\
  &+L^\rho_\ql\Qd^\qa_{\rho\qg}A^\ql_{\qe\qz}-L^\rho_\ql \Qd^\qa_{\rho\qe}A^\ql_{\qg\qz}+L^\rho_\qz\Qd^\ql_{\rho\qe}A^\qa_{\qg\ql}-L^\rho_\qz\Qd^\ql_{\rho\qg}A^\qa_{\ql\qe}\\
  =&\,\nabla_\qg\kk{L^\rho_\qe \Qd^\qa_{\rho\qz}}-\nabla_\qe\kk{L^\rho_\qg \Qd^\qa_{\rho\qz}}.
\end{align*} 
The theorem is proved.
\end{proof}

\begin{theorem}
Given the geometric data $(L,\nabla,\nabla^*)$ where $\nabla$ and $\nabla^*$ are flat torsionless affine connections and $L$ is a $(1,1)$-tensor field with vanishing Nijenhuis torsion, the  following facts are equivalent:
\begin{enumerate}
\item for any bundle isomorphism $\eta^\sharp$, the geometric data $(\eta^\sharp,\nabla)$ and $(L\eta^\sharp,\nabla^*)$ defines a generalised bi-Hamiltonian structure $(\diff{}{\qt_0},\diff{}{\qt_1})$  of hydrodynamic types on the open set where the operator 
  \[
  L_z = L-z I\colon TM\to TM 
  \]
  is invertible for any $z$; 
  \item the Gauss-Manin connection as defined by \eqref{GMnoindex} associated with  $(L,\nabla,\nabla^*)$ is flat;
  \item the pair of  differentials $(d_{\nabla},d_{L\nabla^*})$ define a differential bicomplex on the space of vector valued differential  forms;
  \item the conditions (2) and (3) in Proposition \ref{BY} are satisfied.
\end{enumerate}
\end{theorem}
\begin{proof}
  The equivalence between (2) and (3) is given in \cite{AL_GM}. It follows from Theorem \ref{BZ} that (2) and (4) are equivalent, and from Proposition \ref{BY} that (1) implies (4).
  
  We need only to prove that (2) implies (1). Let us fix any local coordinates $(v^\qa;\qth_\qa)$ of $\hat M$. Denote by $A^\qa_{\qb\qg}$ and $B^\qa_{\qb\qg}$ the Christoffel symbols of connections $\nabla$ and $\nabla^*$ respectively. Let us construct the odd derivatives 
\[
\diff{v^\qa}{\qt_0} = \eta^{\qa\qb}\qth_\qb^1+\Qg^{\qa\qb}_\qg v^\qg_x\qth_\qb,\quad \diff{v^\qa}{\qt_1} = g^{\qa\qb}\qth_\qb^1+\tilde \Qg^{\qa\qb}_\qg v^\qg_x\qth_\qb,
\]
where we have (see \eqref{BC})
\[
\Qg^{\qa\qb}_\qg = \qp_\qg\eta^{\qa\qb}+\eta^{\mu\qb}A^\qa_{\qg\mu},\quad \tilde \Qg^{\qa\qb}_\qg = \qp_\qg g^{\qa\qb}+g^{\mu\qb}B^\qa_{\qg\mu},\quad L^\qa_\ql\eta^{\ql\qb} = g^{\qa\qb}.
\]
We omit the definitions of $\diff{\qth_\qa}{\qt_0}$ and $\diff{\qth_\qa}{\qt_1}$ which are uniquely determined by using the identities \eqref{AL} and \eqref{AZ}. It follows from the definition \eqref{AM} that the associated geometric data of $\diff{}{\qt_0}$ and $\diff{}{\qt_1}$ is $(\eta^\sharp,\nabla)$ and $(L\eta^\sharp,\nabla^*)$.
In order to prove that they define a generalised bi-Hamiltonian structure, we consider 
\begin{align*}
\diff{v^\qa}{\qt_1}-z\diff{v^\qa}{\qt_0} =&\,\left(g^{\qa\qb}-z\eta^{\qa\qb}\right)\qth_\qb^1+ \kk{\tilde\Qg^{\qa\qb}_\qg-z \Qg^{\qa\qb}_\qg} v^\qg_x\qth_\qb.
\end{align*}
Denote by $g_z^\sharp = g^\sharp-z\eta^\sharp = L_z\,\eta^\sharp$, then it is invertible on the open set where $L_z$ is invertible, and its inverse is denoted by $(g_{z\,\qa\qb})$ satisfying
\[
g_z^{\mu\qa}g_{z\,\mu\qb} = \qd^\qa_\qb,\quad g_{z\,\qa\mu}g_z^{\qb\mu} = \qd^\qb_\qa.
\]
It follows that 
\begin{align*}
  \tilde\Qg^{\qa\qb}_\qg-z \Qg^{\qa\qb}_\qg = &\,\qp_\qg g^{\qa\qb}_z+g^{\mu\qb}B^\qa_{\qg\mu}-z\eta^{\mu\qb}A^\qa_{\qg\mu}\\
  =&\,\qp_\qg g^{\qa\qb}_z+g^{\ql\qb}_zB^\qa_{\qg\ql}+z\eta^{\mu\qb}\kk{B^\qa_{\qg\mu}-A^\qa_{\qg\mu}}\\
  =&\,\qp_\qg g^{\qa\qb}_z+g^{\ql\qb}_z\kk{B^\qa_{\qg\ql}+z g_{z\,\ql\rho}\eta^{\mu\rho}\kk{B^\qa_{\qg\mu}-A^\qa_{\qg\mu}}}.
\end{align*}
Using the relation 
\[
g_z^{\qa\qb} = (L_z)^\qa_\rho \eta^{\rho\qb},
\]
we find 
\[
(L_z^{-1})^\mu_\ql =  g_{z\,\ql\rho}\eta^{\mu\rho},
\]
and it follows from the identity \eqref{BC} that the geometric data associated to $\diff{}{\qt_1}-z\diff{}{\qt_0}$ is given by $(g_z^\sharp,\nabla(z))$, where the Christoffel symbols of $\nabla(z)$ read 
\[
\Qg(z)^{\qa}_{\qg\ql} = B^\qa_{\qg\ql}+z (L_z^{-1})^\mu_\ql\kk{B^\qa_{\qg\mu}-A^\qa_{\qg\mu}}.
\]
By comparing this to the expression \eqref{GM}, we conclude that $\nabla(z)$ is exactly the Gauss-Manin connections defined by the geometric data $(L,\nabla,\nabla^*)$. We know from Theorem \ref{GMth} that $\nabla(z)$ is torsionless and flat, and therefore $(\diff{}{\qt_0},\diff{}{\qt_1})$ is a generalised bi-Hamiltonian structure of hydrodynamic type.
\end{proof}

The above theorem gives a clear description of the differential geometry of generalised bi-Hamiltonian structures of hydrodynamic type. Finally, we can relate them to bi-flat F-manifolds.
\begin{theorem}
  Let $(M,\nabla,\circ,e,\nabla^*,*,E)$ be a bi-flat F-manifold and fix an arbitrary bundle isomorphism $\eta^\sharp\colon T^*M\to TM$. Then the geometric data $(\eta^\sharp,\nabla)$ and $(g^\sharp,\nabla^*+\nu *)$ defines a generalised bi-Hamiltonian structure $(\diff{}{\qt_0},\diff{}{\qt_1})$  of hydrodynamic types on the open set where the operator 
  \[
  L_z = E\circ-z I\colon TM\to TM 
  \]
  is invertible for any $z$. Here the bundle morphism $g^\sharp$ is defined by 
  \[
  g^\sharp(\qo) = \kk{E\circ} \eta^\sharp(\qo),\quad\forall\,\qo\in\Qg(TM),
  \]
  and $\nu$ is an arbitrary constant. Moreover, the principal hierarchy is a generalised bi-Hamiltonian system with respect to $(\diff{}{\qt_0},\diff{}{\qt_1})$.
\end{theorem}
\begin{proof}
The first part of the theorem follows from the previous theorem taking into account that
 the curvature of Gauss-Manin connections associated with bi-flat F-manifolds vanishes.

For the principal hierarchy associated with $M$, we already know from the proof of Theorem \ref{BD} that a flow in the principal hierarchy corresponds to a vector field $X$ such that $d_\nabla(X\circ) = 0$, and it only remains to show that $d_{\nabla^*+\nu *}(X\circ) = 0$. Pick two vector fields $Y$ and $Z$, we conclude from the definition \eqref{AU} of the exterior covariant differential that  
\begin{align*}
  d_{\nabla^*+\nu *}(X\circ)(Y,Z) =&\, \nabla^*_Y(X\circ Z)+\nu Y*(X\circ Z)\\
  &-\nabla^*_Z(X\circ Y)-\nu Z*(X\circ Y)-X\circ([Y,Z])\\
  =&\, d_{\nabla^*}(X\circ)(Y,Z)+\nu (E\circ)^{-1}\kk{Y\circ(X\circ Z)}\\
  &-\nu (E\circ)^{-1}\kk{Z\circ(X\circ Y)}\\
  =&\,0,
\end{align*}
here we use the fact \eqref{comp_conn} that $d_{\nabla^*}(X\circ) =d_{\nabla}(X\circ) = 0 $ and use the definition \eqref{dual} of the dual product. The theorem is proved.
\end{proof}

This theorem generalises the well-known result of the bi-Hamiltonian property of the principal hierarchies of Dubrovin-Frobenius manifolds.

\subsection{Canonical form of geometric data for semisimple bi-flat F-manifolds}   
It is known (see \cite{ALmulti}) that semisimple bi-flat F-manifolds $(M,\nabla,\circ,e,\nabla^*,*,E)$, 
 in canonical  coordinates $(u^1,...,u^n)$,  are parametrized by the solutions of the system
\begin{eqnarray*}
&&\d_kA^i_{ij}=-A^i_{ij}A^i_{ik}+A^i_{ij}A^j_{jk}
+A^i_{ik}A^k_{kj}, \quad i\ne k\ne j\ne i,\\
&&e(A^i_{ij})=0,\qquad i\ne j\\
&&E(A^i_{ij})=-A^i_{ij},\qquad i\ne  j
\end{eqnarray*}
Given a solution of the above system the bi-flat structure is determined by the following formulas
\begin{itemize}
\item the connection $\nabla$ is defined by the Christoffel symbols $A^i_{ij}$ and
\begin{equation*}\label{naturalc}
\begin{split}
A^i_{jk}&=0\qquad\forall i\ne j\ne k \ne i\\
A^i_{jj}&:=-A^i_{ij}\qquad i\ne j\\
A^i_{ii}&:=-\sum_{l\ne i}A^i_{li},
\end{split}
\end{equation*} 
\item the dual connection $\nabla^*$ defined by is defined by the Christoffel symbols $B^i_{ij}=A^i_{ij}$ for $i\neq j$ and
\begin{equation*}\label{dualnabla}
\begin{split}
B^i_{jk}&:=0\qquad\forall i\ne j\ne k \ne i\\
B^i_{jj}&:=-\f{u^i}{u^j}A^i_{ij}\qquad i\ne j\\
B^i_{ii}&:=-\sum_{l\ne i}\f{u^l}{u^i}A^i_{li}-\f{1}{u^i},
\end{split}
\end{equation*}
\item the structure constants of  $\circ$ defined by $c^i_{jk}=\delta^i_j\delta^i_k$,
\item the structure constants  of $*$ defined by $c^{*i}_{jk}=\f{1}{u^i}\delta^i_j\delta^i_k$, 
\item the vector fields $e=\sum_{i=1}^n\frac{\partial}{\partial u^i}$ and $E=\sum_{i=1}^nu^i\frac{\partial}{\partial u^i}$. 
\end{itemize} 
In these coordinates, it is easy to see that we have $L=E\circ=\text{diag}(u^1,...,u^n)$.
\newline
\newline
We call the previous formulas for $(L,\nabla,\nabla^*)$ the canonical form of geometric data for generalised bi-Hamiltonian structures associated with bi-flat F-manifolds.

\section{Conclusions}
In this paper, we propose the notions of generalised (bi)-Hamiltonian structures as well as generalised (bi-)Hamiltonian integrable hierarchies. They can be naturally associated to (bi-)flat F-manifolds and enable a further study towards understanding the corresponding principal hierarchies. These notions admit clear differential geometric descriptions, and provide deeper insight to geometry of (bi-)flat F-manifolds.
We think that the results of this paper might have significant consequences in the theory of integrable systems, in particular in the  Dubrovin-Zhang axiomatic framework \cite{DZ01}.

We propose here some of the important problems regarding the generalised Hamiltonian structure.
\begin{itemize} 
\item The classification of generalised (bi-)Hamiltonian structures as well as generalised (bi-)Hamiltonian integrable hierarchies. The classification should be done by constructing and computing certain cohomology groups. The key difference from the classification of usual (bi-)Hamiltonian structure is that we have a much larger automorphism group in the generalised setting. Indeed, as we already see in this paper, we allow general transformations of odd variables which do not exist in the usual setting of Hamiltonian structures.
 \item The classification of generalised bi-Hamiltonian structure enables us to study bi-flat F-manifolds in higher genus. We expect that there exist axioms for ``generalised bi-Hamiltonian integrable hierarchies of topological type'', such that these hierarchies coincide with those considered in \cite{ABLR1}. This will provide a generalisation of Dubrovin-Zhang-Givental-Teleman equivalence for semisimple CohFT.
\item The double ramification hierarchy associated with a homogeneous CohFT is bi-Hamiltonian. An explicit formula for the bi-Hamiltonian structure (and thus the bi-Hamiltonian recursion) has been conjectured in \cite{BRS} and proved in \cite{BR}. In \cite{ABLR2} it was conjectured that a similar recursion holds true also in the case (which in general is not even Hamiltonian) of homogeneous F-CohFT. The results of the present paper suggest that this recursion might admit an interpretation in terms of the generalised bi-Hamiltonian structure obtained by deforming the generalised bi-Hamiltonian structure of hydrodynamic type associated with the underlying bi-flat F-manifold. 

\item The study of generalised Hamiltonian structure also has application in study of usual Hamiltonian integrable hierarchies. Indeed, in \cite{LWZ4}, we already see that generalised bi-Hamiltonian structures arise naturally when we consider linear reciprocal transformation of bi-Hamiltonian structures. We expect that the tools provided in this paper can be used to study the complete classification of bi-Hamiltonian integrable hierarchies under Miura-reciprocal transformations.
 \end{itemize}

\end{document}